\pdfoutput=1
\documentclass[runningheads]{llncs}

\usepackage{amsmath,amssymb,stmaryrd} %
\usepackage{xspace}
\usepackage{xcolor}
\usepackage{url}
\usepackage{listings}
\usepackage{graphicx}
\usepackage{hyperref}
\usepackage{cleveref}
\usepackage{wrapfig}
\usepackage{tabularx}

\definecolor{sh_comment}{rgb}{0.12, 0.38, 0.18}
\definecolor{sh_keyword}{rgb}{0.37, 0.08, 0.25}  %
\definecolor{sh_string}{rgb}{0.06, 0.10, 0.98} %
\def\lstsmallmath{\leavevmode\ifmmode \scriptstyle \else  \fi}
\def\lstsmallmathend{\leavevmode\ifmmode  \else  \fi}
\definecolor{KWColor}{rgb}{0.5,0,0.67}
\definecolor{CommentColor}{rgb}{0.15,0.5,0.15}
\definecolor{lightgrey}{rgb}{0.8,0.8,0.8}

\lstdefinelanguage{JavaScript}[]{Java}{
   morekeywords={var,class,object,function,undefined} 
}

\lstdefinestyle{Eclipse}{
  xleftmargin=0pt,
  language = JavaScript,
  basicstyle=\ttfamily\scriptsize,
  stringstyle=\color{sh_string},
  keywordstyle = \color{sh_keyword}\bfseries,  %
  lineskip=-0.0em,
  commentstyle=\color{sh_comment}\itshape,  
  escapeinside={/*@}{@*/},
  numbersep=5pt,
  captionpos=b,
  xleftmargin=0.4cm, xrightmargin=0.5cm,
   morekeywords={invokestatic,invokeinterface,invokevirtual,invokespecial},
}

\newcommand{\code}[1]{\lstinline[basicstyle=\small\ttfamily]{#1}\xspace}
\newcommand{\Edge}[2]{\mbox{#1}$\rightarrow$\mbox{#2}}

\newcommand{\register}{\textsf{\small{register}}}
\newcommand{\emit}{\textsf{\small{emit}}}
\newcommand{\invoke}{\textsf{\small{invoke}}}
\newcommand{\identity}{\textsf{\small{id}}}

\newcommand{\js}{Java\-Script\xspace}
\newcommand{\eg}{\emph{e.g.}\xspace}
\newcommand{\ie}{\emph{i.e.}\xspace}
\newcommand{\Borges}{\textsc{Borges}\xspace}
\newcommand{\microfun}[4]{\ensuremath{\big\langle #1,\allowbreak #2,\allowbreak #3,\allowbreak #4 \big\rangle}\xspace}

\newcommand{\Gstar}{\ensuremath{G^{*}}}
\newcommand{\Nstar}{\ensuremath{N^{*}}}
\newcommand{\Estar}{\ensuremath{E^{*}}}
\newcommand{\Ghash}{\ensuremath{G^{\#}}}
\newcommand{\Nhash}{\ensuremath{N^{\#}}}
\newcommand{\Ehash}{\ensuremath{E^{\#}}}
\newcommand{\pair}[2]{\ensuremath{\langle #1, #2 \rangle}}
\newcommand{\Zero}{\ensuremath{\textrm{\textbf{0}}}}
\newcommand{\startmain}{\ensuremath{\textsf{start}_\textsf{main}}}
\newcommand{\MF}{\ensuremath{M_{F}}}
\newcommand{\MEnv}{\ensuremath{M_\textsf{Env}}}
\newcommand{\TopEnv}{\ensuremath{\top_\textsf{Env}}}
\newcommand{\SEnv}{\ensuremath{S_\textsf{Env}}}
\newcommand{\VP}{\ensuremath{\textsf{VP}}}
\newcommand{\EdgeFn}{\ensuremath{\textsf{EdgeFn}}}
\newcommand{\env}{\ensuremath{\textsf{env}}}
\newcommand{\Transform}{\ensuremath{\mathcal{T}}}
\newcommand{\Untransform}{\ensuremath{\mathcal{U}}}
\newcommand{\SolveIFDS}{\ensuremath{\textsf{MVP}_\textsf{IFDS}}}
\newcommand{\SolveIDE}{\ensuremath{\textsf{MVP}_\textsf{IDE}}}

\definecolor[named]{ACMPurple}{cmyk}{0.55,1,0,0.15}
\definecolor[named]{ACMDarkBlue}{cmyk}{1,0.58,0,0.21}

\hypersetup{
  bookmarks=true,
  unicode=true,
  pdftitle={Precise Dataflow Analysis of Event-Driven Applications},
  pdfauthor={Ming-Ho Yee,
             Ayaz Badouraly,
             Ond\v{r}ej Lhot\'{a}k,
             Frank Tip,
             Jan Vitek},
  pdfkeywords={static analysis,
               interprocedural analysis,
               asynchronous,
               event-driven},
  colorlinks=true,
  linkcolor=ACMPurple,
  citecolor=ACMPurple,
  urlcolor=ACMDarkBlue,
  filecolor=ACMDarkBlue
}

\begin{document}

\title{Precise Dataflow Analysis of Event-Driven Applications}

\author{Ming-Ho Yee,\inst{1}
Ayaz Badouraly,\inst{2}
Ond\v{r}ej Lhot\'{a}k,\inst{3}
Frank Tip,\inst{1}
Jan Vitek\inst{1,4}}
\authorrunning{M.H. Yee, A. Badouraly, O. Lhot\'{a}k, F. Tip, J. Vitek}
\institute{Northeastern University \and
CentraleSup\'{e}lec \and
University of Waterloo \and
Czech Technical University}

\maketitle

\vspace{-2em}
\begin{abstract}

Event-driven programming is widely used for implementing user interfaces, web applications, 
and non-blocking I/O. An event-driven program is organized as a collection of event handlers 
whose execution is triggered by events. 
Traditional static analysis techniques are unable to reason precisely about event-driven code
because they conservatively assume that event handlers may execute in any order. 
This paper proposes an automatic transformation from Interprocedural Finite
Distributive Subset~(IFDS) problems to Interprocedural Distributed Environment~(IDE)
problems as a general solution to obtain precise static analysis of
event-driven applications; problems in both forms can be solved by existing
implementations. Our contribution is to show how to improve analysis precision by
automatically enriching the former with information about the state of event handlers to 
filter out infeasible paths. We prove the
correctness of our transformation and report on experiments with a proof-of-concept
implementation for a subset of \js.


\end{abstract}
 \vspace{-2em}
\section{Introduction}\label{sec:intro} 

Event-driven programming is a popular paradigm in which control
flow follows the order of events.
The essence of the paradigm is the flexible association
between user-defined event handlers and events, such as
user interface or operating system actions.
When an event is emitted, all 
event handlers that have been registered for it are eligible to be invoked by
the event loop.

Flexibility comes from the fact that event handlers are invoked
asynchronously.
This asynchrony causes complexity in reasoning about event-driven
programs in the presence of mutable state: consider the example of
a global variable initialized by one event handler and used by
another. The order in which the event handlers are invoked is
critical for correctness, but the ordering constraints are not explicit;
responsibility for the ordering is imposed on the programmer.

To reason about event-driven programs, a static analysis must
model the execution of the event loop. A conservative---but imprecise---approach
is to assume that any handler can be invoked in any order, ignoring any run-time
constraints. Work by Madsen et
al.~\cite{Madsen15} avoids such imprecision by using a notion of
context sensitivity in which a context abstracts the set of
event handlers registered and the set of events emitted.
The resulting context-sensitive call graphs can distinguish,
\eg, program states where no events
have been emitted and program states where an event has been emitted,
resulting in a more precise analysis of event-driven
programs. Unfortunately, the number of contexts is exponential in the size of the program,
so the analysis does not scale.

We propose a technique to write static analysis algorithms without considering
the ordering of events and registrations, and then translate them automatically
into algorithms that filter out infeasible paths.
We leverage two established static analysis frameworks, the
\emph{Interprocedural Finite Distributive Subset}~(IFDS) framework
introduced by Reps et al.~\cite{Reps95} and the \emph{Interprocedural
  Distributive Environment}~(IDE) framework of Sagiv et
al.~\cite{Sagiv96}. These frameworks have been used on a variety of
practical problems, including taint
analysis~\cite{DBLP:conf/pldi/ArztRFBBKTOM14}, and a number of solvers are
available~\cite{DBLP:conf/pldi/ArztRFBBKTOM14,Heros,WALA,Madsen16}.

The IFDS framework solves \emph{interprocedural} dataflow problems
whose domain consists of \emph{subsets} of a \emph{finite} set $D$, and
whose dataflow functions are \emph{distributive}, and it computes
a meet-over-valid-paths solution in polynomial time.
Any static analysis
that can be expressed in this framework is a candidate for our
approach.
Unfortunately, IFDS cannot enforce constraints on the execution order
of event handlers.  To overcome this limitation, our approach automatically
translates an arbitrary IFDS analysis into an IDE analysis.

The IDE framework generalizes IFDS by using \emph{environments} as dataflow
facts, \ie, maps from some finite set $D$ to some
lattice of values $L$, and \emph{distributive} \emph{environment transformers} as dataflow
functions.  Like IFDS, IDE problems can be solved efficiently.
If the IFDS
algorithm computes facts in $D$ that hold along interprocedurally valid
paths, then the IDE algorithm computes values from $L$ along those paths.
Our approach associates dataflow functions to edges
associated with events and event handlers,
so that the composed transfer functions
filter out dataflow facts reachable only along infeasible paths.

Our main contribution is an automated transformation from IFDS into
IDE problems, such that the IDE result solves the original
IFDS problem but avoids
imprecision due to infeasible paths.
We prove our transformation sound and precise.
We demonstrate a proof-of-concept tool called
\Borges,
which is capable of
analyzing small programs in a subset of \js that use event-driven
programming.
We report on three case studies 
on small Node.js
programs that use events for asynchronous file I/O, timers, and network
I/O. We demonstrate precision improvements
in an IFDS-based possibly uninitialized variables dataflow analysis.
Our technique is applicable to other frameworks and languages.
\section{Motivating Examples}\label{sec:motivation}

\begin{figure}[!t]
  \centering
  \begin{minipage}[c]{0.51\textwidth}
    {\scriptsize 
\begin{lstlisting}
/*@
% Keep this in sync with dirstat-clean.js, the copy is to get an accurate file size
@*/var fs = require('fs'); /*@\label{line:RequireFs}@*/
var sum; /*@\label{line:declare_sum}@*/
fs.readdir('.', function f(err, files) {/*@\label{line:readdir}@*/
    if (err) throw err; /*@\label{line:checkForErrors1}@*/
    sum = 0; /*@\label{line:initializeSum}@*/
    files.forEach(function g(file) { /*@\label{line:ForEach}@*/
        fs.stat('./' + file, /*@\label{line:stat}@*/
          function h(err, stats) {
            if (err) throw err; /*@\label{line:checkForErrors2}@*/
            var sz = stats.size; /*@\label{line:readFileSize}@*/
            sum += sz; /*@\label{line:addToSum}@*/
            console.log(file + ' ' + sz);/*@\label{line:FsPrint1}@*/
            console.log('sum ' + sum);  /*@\label{line:FsPrint2}@*/
          });
    });
});
console.log('done'); /*@\label{line:FsDone}@*/
\end{lstlisting}}
  \end{minipage}
  \begin{minipage}[c]{0.48\textwidth}
    \includegraphics[width=\textwidth,trim={0.8cm 0.3cm 0.8cm 0.5cm},clip]{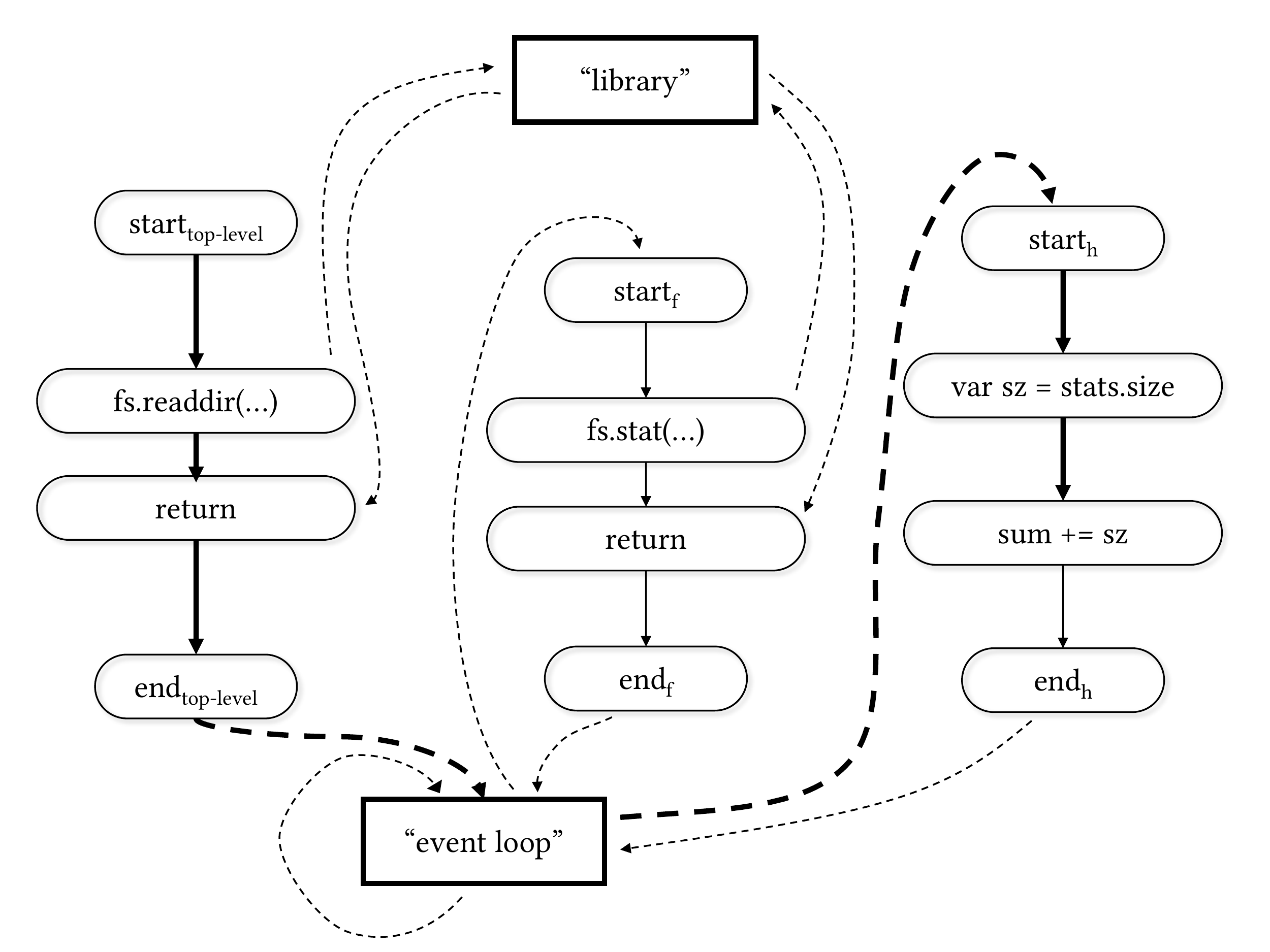}
  \end{minipage}
  \vspace{-5pt}
  \caption{Example application \code{dirstat.js} and its control-flow
    supergraph. Interprocedural edges are dashed; an infeasible path is shown in
    bold.
  We treat top-level code as if it occurs inside a function
\code{top-level}. To avoid clutter, library code is represented using a single
node labeled ``library'' and further details that have no bearing on the topic
have been elided.
  }\label{fig:dirstat}
  \vspace{-10pt}
\end{figure}

\Cref{fig:dirstat} shows an event-driven \js application that uses the
Node.js \code{fs} (File System)
module.
Running the application
prints the names and sizes of the files in the current
directory, as well as a running sum of their sizes.

We briefly discuss the workings of the application. First, the \code{fs} module
is loaded (line~\ref{line:RequireFs}), making various file-related operations
available as methods on an object assigned to variable~\code{fs}. Next,
variable \code{sum} is declared, but not initialized
(line~\ref{line:declare_sum}). Line~\ref{line:readdir} calls \code{readdir} to read
the contents of the current directory, with two
arguments: a path to the directory that is to be read and a callback function, \code{f}.
\code{f} is asynchronously invoked with two arguments,
\code{err} and \code{files}, where \code{err} is either \code{null} or \code{undefined}
if the operation completes successfully or an error object otherwise, and
\code{files} is an array containing the names of the files in the directory.

When \code{f} is invoked, it checks if an error occurred
(line~\ref{line:checkForErrors1}). If not, it initializes \code{sum} to \code{0}
(line~\ref{line:initializeSum}), and uses the built-in
\code{forEach} function to iterate through all names in array \code{files}
(line~\ref{line:ForEach}).  \code{forEach} takes a callback,
\code{g}, that is invoked synchronously for each array element, binding it to variable \code{file}. 
For each file name, the function \code{stat} is invoked to access some properties of that file
(line~\ref{line:stat}). The second argument passed to \code{stat} is a callback,
\code{h}, that is asynchronously invoked with two arguments, \code{err} and
\code{stats}, where \code{stats} is an object containing information about
the current file. When \code{h} is invoked, it retrieves the
size of this file, stores it in variable \code{sz}
(line~\ref{line:readFileSize}), and adds it to \code{sum}
(line~\ref{line:addToSum}). Then, it prints information about the current file
(lines~\ref{line:FsPrint1}--\ref{line:FsPrint2}). Lastly, the application prints `\code{done}'
(line~\ref{line:FsDone}).

\medskip\noindent\textbf{Execution behavior.}
Executing the program in a directory
containing, in addition to the script itself, a file \code{f1} of size 100 and a
file \code{f2} of size 50, prints:
{\small
\begin{verbatim}
done\n dirstat.js 428\n sum 428\n f1 100\n sum 528\n f2 50\n sum 578\n
\end{verbatim}
}
\noindent
Note that `\code{done}' is printed first, because the callback \code{f} 
registered by \code{readdir}  does not execute until after the top-level
code has finished executing.

\medskip\noindent\textbf{Representing asynchronous control flow.}
The callbacks passed to
\texttt{\small readdir}
and \code{stat} are invoked \textit{asynchronously}. Since \js's execution model is
single-threaded and non-preemptive, these functions will not 
execute until the current callback has finished executing.
\Cref{fig:dirstat}
shows the interprocedural control flow graph~(ICFG) for the application.
An ICFG (also known as a \emph{supergraph} in the IFDS literature) contains a subgraph for each function in the application,
 with nodes for all expressions in the function and edges reflecting
possible control flow between them. Each such subgraph contains distinct ``start'' and ``end''
nodes representing the function's entry and exit points. Edges between subgraphs
represent interprocedural control flow between functions due to calls and returns. 
Asynchronous control flow is modeled by way of a special ``event loop'' node. 
Edges connect each function's end node to the event loop node, reflecting
that control returns to the event loop when a function at the top of the call stack finishes executing.
Edges connect the ``event loop'' node to the ``start'' node for each asynchronously
invoked function. Thus, in \cref{fig:dirstat}, there are edges from
``event loop'' to the start nodes for \code{f} and \code{h}. 
 
\medskip\noindent\textbf{Static analysis.}
Suppose that we want to perform a dataflow analysis to determine potentially uninitialized variables. 
This problem can be expressed in terms of a domain consisting of \emph{subsets}
of a \emph{finite} set $D$ (in this example, the set of possibly uninitialized
variables),
and using dataflow functions that are \emph{distributive}, so a meet-over-valid-paths solution 
can be computed in polynomial time using the IFDS framework~\cite{Reps95}. 
The defining characteristic of IFDS is that it avoids imprecision that would arise from considering
data flow along control-flow paths in which function calls and function returns are not matched up properly. 

However, suppose the analysis considers the control-flow path shown in bold in
\cref{fig:dirstat}, where execution of top-level code is followed by execution
of \code{h}, without ever calling \code{f}.
On this path,  \code{sum} is referenced on line~\ref{line:addToSum} without 
having been initialized, so a traditional IFDS-based analysis will report that \code{sum} is possibly uninitialized 
on line~\ref{line:addToSum}. In reality, this path is \emph{infeasible} because \code{h} cannot be invoked asynchronously 
before being registered during execution of \code{f}.
Furthermore,  since \code{f} initializes \code{sum} and registers callback \code{h} (recall that
\code{g} is invoked synchronously by \code{forEach}), and \code{h} cannot be
invoked until after \code{f} has finished executing, \code{sum} is guaranteed to be initialized when \code{h} executes.

This paper presents a technique for improving the precision of IFDS-based analyses by taking into account 
the order in which callbacks can execute. Our approach involves transforming the original IFDS problem 
into an IDE problem~\cite{Sagiv96}
by associating dataflow functions with edges corresponding to \emph{event handler registration} and 
\emph{event handler invocation}. The transfer function obtained by composing the functions along a control-flow path 
reflects that path's feasibility, thus effectively ``filtering out'' dataflow facts if the path is infeasible.   

\begin{figure}[!t]
  \centering
  \begin{minipage}[c]{0.42\textwidth}
    {\scriptsize 
\begin{lstlisting}
function hdlClose(){ /*@\label{line:HandleCloseStart}@*/
  txt = txt.concat(', world!');/*@\label{line:CallConcat}@*/
  console.log(txt); /*@\label{line:Print}@*/
} /*@\label{line:HandleCloseEnd}@*/

function hdlOpen(){ /*@\label{line:HandleOpenStart}@*/
  txt = 'Hello'; /*@\label{line:AssignToText}@*/
  door.on('close', hdlClose);  /*@\label{line:RegisterHandleClose}@*/
  door.emit('close')  /*@\label{line:EmitClose}@*/
} /*@\label{line:HandleOpenEnd}@*/

var e = require('events'); /*@\label{line:Start}@*/
var door = new e.EventEmitter();/*@\label{line:CreateDoor}@*/
var txt; /*@\label{line:DeclareText}@*/

door.on('open', hdlOpen); /*@\label{line:RegisterHandleOpen}@*/
door.emit('open') /*@\label{line:EmitOpen}@*/
\end{lstlisting}}
  \end{minipage}
  \begin{minipage}[c]{0.57\textwidth}
    \includegraphics[width=0.9\textwidth,trim={1cm 1.5cm 0.7cm 2.4cm},clip]{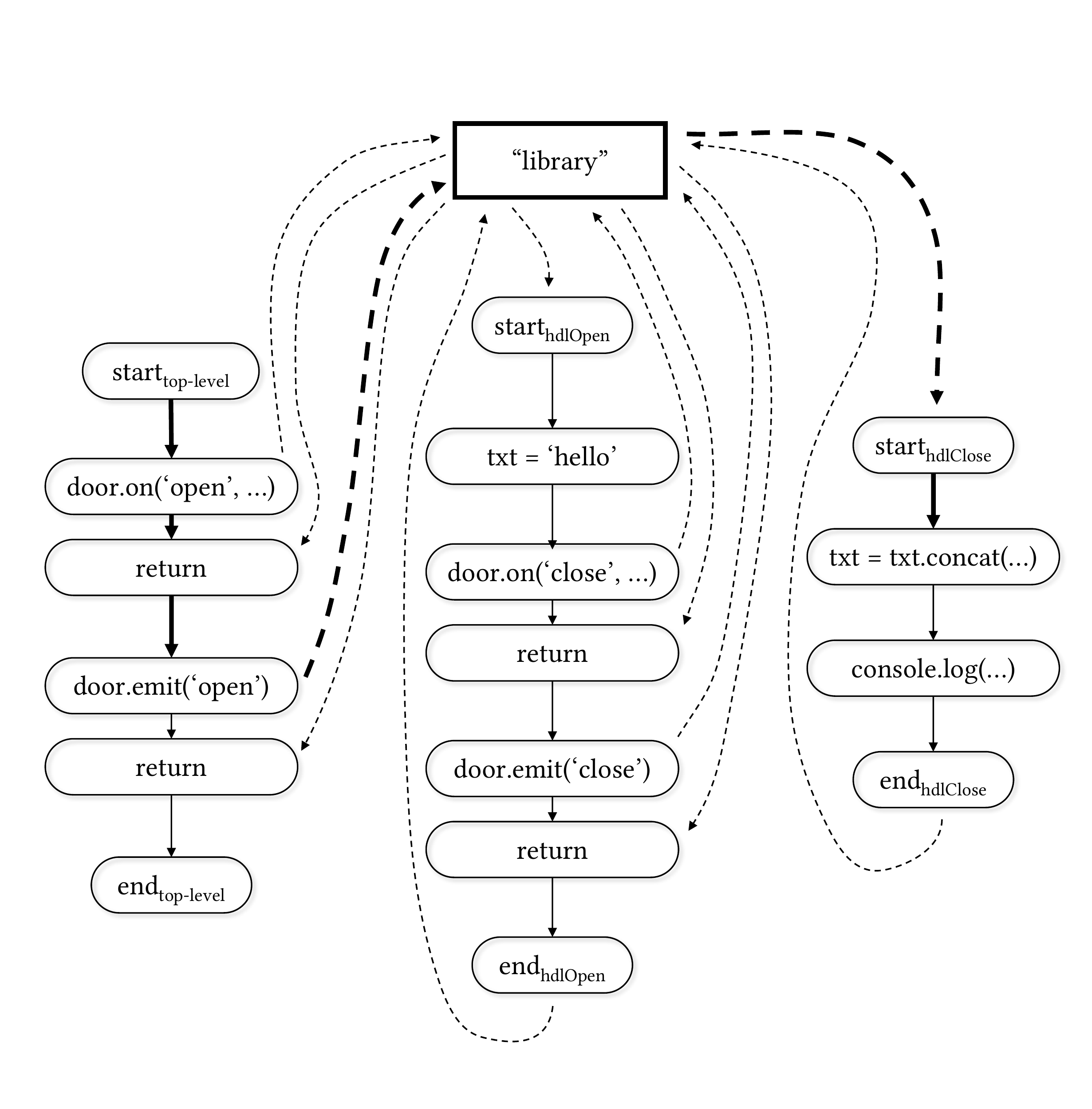}
  \end{minipage}
  \caption{An example illustrating an infeasible path during analysis.
  In the supergraph, interprocedural edges are dashed and an infeasible path is
shown in bold.}\label{fig:DoorsExample}
  \vspace{-10pt}
\end{figure}

\medskip\noindent\textbf{Explicit emission of events.}
\Cref{fig:DoorsExample} illustrates a more complex scenario where the \code{EventEmitter} class of the Node.js \code{events} 
package
is used to model a door that responds to \code{open} and \code{close} events.
On line~\ref{line:RegisterHandleOpen}, function \code{hdlOpen} is registered to handle the \code{open} event on \code{door}, and on
line~\ref{line:RegisterHandleClose}, \code{hdlClose} is registered to handle the \code{close} event.
To trigger event handlers, an event must be emitted using the \code{emit} method.

We consider the program's execution behavior.
After loading the \code{events} package (line~\ref{line:Start}), the program creates a door (line~\ref{line:CreateDoor})
and declares variable \code{txt} (line~\ref{line:DeclareText}).
The call \code{door.on(...)}  (line~\ref{line:RegisterHandleOpen}) associates
\code{hdlOpen} with the \code{open} event.
Calling \code{emit} triggers \code{hdlOpen},%
\footnote{
 \js is single-threaded and non-preemptive. \code{emit} yields control to the event loop, which invokes 
  the associated handler, and control returns to the caller of \code{emit}.} 
which, when it executes, initializes \code{txt} to \code{'Hello'} (line~\ref{line:AssignToText}) and
 associates \code{hdlClose} with the \code{close} event (line~\ref{line:RegisterHandleClose}).
Line~\ref{line:EmitClose} emits the \code{close} event, triggering its handler,
\code{hdlClose}, which,
when it executes,  updates \code{txt} (line~\ref{line:CallConcat}) and prints its value \code{'Hello, world!'}.
Note that \code{hdlClose} must execute \emph{after} \code{hdlOpen}, because it responds only to the \code{close} event, 
which is emitted in the body of \code{hdlOpen}.

In the ICFG, several call sites invoke library functions such as
\code{on} and \code{emit}, while the library invokes \code{hdlOpen} and
\code{hdlClose}.
No ordering exists between the \Edge{\code{library}}{\code{hdlOpen}} and
\Edge{\code{library}}{\code{hdlClose}} edges, 
so a traditional analysis assumes that these event handlers may execute in an arbitrary order. In particular, the path shown in bold 
is admitted, but it is infeasible because it entails \code{hdlClose} executing before \code{close} is emitted.

To understand the impact of imprecision, we again consider an analysis that looks for uninitialized variables.
If the analysis considers the infeasible path, it concludes that
\code{txt.concat(...)} may take place at a time when \code{txt} is uninitialized.
This is a false positive because it is impossible for \code{hdlClose} to execute before being registered or before the \code{close} event is emitted.
 
For this example, we would like to rule out the path marked in bold by tracking 
three operations associated with each event handler:
 (i) when an event handler is \emph{registered} for an event,
 (ii) when the event is \emph{emitted}, and
 (iii) when the event handler is \emph{invoked}.
Infeasible paths will be filtered out if operation (i) does not happen before operation (iii), and if operation 
(ii) does not happen before operation (iii). 
To do so, we will determine the possible sequences of these operations associated with each dataflow fact, and 
filter out those dataflow facts associated with infeasible sequences.
Note that in the file system example discussed previously, \code{emit} operations are not explicitly present in the application source code, so it can
be viewed as a special case of the more general scenario discussed here.

\section{Background}\label{sec:background}

Our technique takes as input an instance of the IFDS framework and
outputs an instance of the IDE framework. In this section, we provide
some background about these frameworks.

\medskip\noindent\textbf{IFDS background.}
The IFDS framework~\cite{Reps95} is applicable to \emph{interprocedural}
dataflow problems whose domain consists of \emph{subsets} of a \emph{finite} set
$D$, and whose dataflow functions are \emph{distributive} (\ie,
$f$ is distributive if and only if $f(x_1 \sqcap x_2) = f(x_1) \sqcap f(x_2)$).
It has proven to be sufficiently expressive and efficient to accommodate
classical dataflow problems such as the possibly uninitialized variables problem
illustrated in \cref{fig:DoorsExample}, but also more complex problems such
as taint analysis~\cite{DBLP:conf/pldi/ArztRFBBKTOM14} and typestate
analysis~\cite{Fink:2008:ETV:1348250.1348255,DBLP:conf/oopsla/NaeemL08}.

An IFDS problem instance $P$ is defined as $\langle \Gstar, D, F, \MF,
\sqcap \rangle$, where:

\begin{enumerate}
  \item $\Gstar = \pair{\Nstar}{\Estar}$ is the ICFG
    of the input program, called the supergraph;
  \item $D$ is a finite set of dataflow facts;
  \item $F \subseteq 2^{D} \to 2^{D}$ is a set of distributive dataflow
    functions;
  \item $\MF : \Estar \to F$ maps supergraph edges to dataflow functions;
    and
  \item $\sqcap$ is the meet operator on the powerset $2^D$ (either union or
    intersection).
\end{enumerate}

\noindent The IFDS framework computes in polynomial time the
meet-over-valid-paths solution,%
\footnote{Following the IFDS and IDE literature, throughout
  this paper, we use the lattice \emph{meet} operation, $\sqcap$, to merge
  dataflow facts when control-flow paths merge. Thus the top element,~$\top$, of
  a lattice represents an unreachable state and the bottom element,~$\bot$,
  means that all concrete states are possible.}
\mbox{$\SolveIFDS : \Nstar \to 2^D$}, of the
dataflow constraints, where each node~\mbox{$n \in \Nstar$} is mapped to a set
of dataflow facts.
A valid path respects the fact
  that, when a function finishes executing, it returns to the call site from
  where it was invoked. $\VP(n)$ denotes the set of all valid paths
  from the start of the program to node $n$. Formally, the meet-over-valid-paths
  solution is defined as
\begin{equation*}\label{eqn:solve-ifds}
  \SolveIFDS(P) = \lambda n . \!\bigsqcap_{p \,\in\, \VP(n)} \!\!\MF(p)(\varnothing)
\end{equation*}
where $\MF$ is extended to paths so that $\MF([e_1 \ldots e_k]) =
\MF(e_k) \circ \cdots \circ \MF(e_1) \circ \textsf{id}$.

\begin{figure}[!t]
  \centering
  \begin{minipage}[b]{0.48\textwidth}
    \begin{align*}
      f = \lambda S . & \textbf{ if }   y \in S \vee z \in S \\
                      & \textbf{ then } S \cup \{ x \} \\
                      & \textbf{ else } S \setminus \{ x \}
    \end{align*}
    \scriptsize
    $R_f = \big\{ \pair{\Zero}{\Zero}, \pair{y}{x}, \pair{y}{y}, \pair{z}{x},
    \pair{z}{z} \big\}$\\
  \end{minipage}
  ~ %
  \begin{minipage}[b]{0.48\textwidth}
    \centering
    \includegraphics[width=0.7\textwidth,trim={4cm 2.7cm 4cm 3.5cm},
    clip]{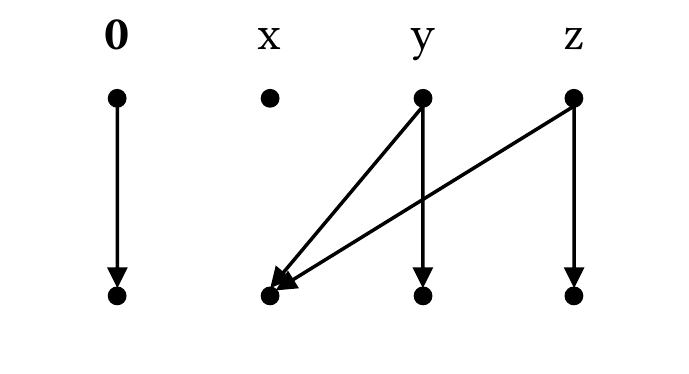}
  \end{minipage}
  \vspace{-10pt}
  \caption{
    Representing the effect of \code{x = y + z} for the
    possibly uninitialized variables analysis. The dataflow function is
    \emph{distributive}, as only one input needs to be considered at a time:
    \code{x} is possibly uninitialized if \code{y} \emph{or} \code{z} are
    possibly uninitialized.}\label{fig:RepresentationRelation}
    \vspace{-10pt}
\end{figure}

The key insight 
behind the IFDS algorithm is that any distributive function $f : 2^D \to
2^D$ can be represented as a bipartite graph with $2(D + 1)$ nodes, with edges
from one instance of $D \cup \{\Zero\}$ to another instance of $D \cup
\{\Zero\}$; \cref{fig:RepresentationRelation} illustrates an example.
Formally, the \emph{representation relation}, $R_f \subseteq (D \cup
\{\Zero\}) \times (D \cup \{\Zero\})$, of a distributive function $f : 2^D \to
2^D$, is defined as follows:
\begin{equation*}\label{eqn:rep-rel}
  R_f = \{ \pair{\Zero}{\Zero} \} \cup
        \{ \pair{\Zero}{d} \,|\, d \in f(\varnothing) \} \cup
        \{ \pair{d_1}{d_2} \,|\, d_2 \in f(\{d_1\}) \wedge d_2 \notin f(\varnothing) \}.
\end{equation*}
The edges of the representation relation are sufficient to uniquely
determine~$f(D_0)$ for any subset $D_0 \subseteq D$, since by distributivity
\mbox{$f(D_0) = f(\varnothing) \sqcap \bigsqcap_{d \in D_0} f(\{d\})$}.
Also, the meet and composition of two distributive
functions $f, g \in 2^D \to 2^D$ can be computed and represented as bipartite
graphs, as shown in \cref{fig:RepRelMeetCompose}:
\begin{align*}
  R_{f \sqcap g} &= \{ \pair{x}{y} \,|\, \pair{x}{y} \in
                       R_f \cup R_g \} \\
  R_{g \circ f}  &= \{ \pair{x}{z} \,|\,
                       \exists y \in D \cup \{\Zero\} \,.\,
                       \pair{x}{y} \in R_f \wedge
                       \pair{y}{z} \in R_g \}.
\end{align*}

IFDS represents a given problem instance $P = \langle \Gstar, D,
F, \MF, \sqcap \rangle$ as an \emph{exploded supergraph}, $\Ghash_{P} =
\pair{\Nhash}{\Ehash}$, where:

\begin{enumerate}
  \item $\Nhash = \Nstar \times (D \cup \{\Zero\})$, and
  \item $\Ehash = \{ \pair{m}{d_1} \to \pair{n}{d_2} \,|\,
    \pair{m}{n} \in \Estar \wedge \pair{d_1}{d_2} \in R_{\MF(m
    \to n)} \}$.
\end{enumerate}

\noindent In essence, each node $n \in \Nstar$ of the supergraph has been
``exploded'' into a set of nodes $\pair{n}{d}$, where each $d$ is a dataflow
fact (or~$\Zero$), and each edge $e \in \Estar$ becomes the set of edges from
the representation relation $R_{\MF(e)}$, where $\MF(e)$ is the dataflow
function assigned to $e$. In this graph, a node $\pair{n}{d}$ is reachable from
the start node $\pair{\startmain}{\Zero}$ if and only if fact $d$
holds at statement $n$.

The algorithm works by iteratively composing a dataflow function for an existing
control-flow path with the dataflow function for an additional instruction, thus
yielding a dataflow function for a longer path. Once a
path covers an entire procedure, its dataflow function becomes a summary
function for the procedure and is used to model the effect of the procedure at
its call sites.

\begin{figure}[!t]
  \centering
  \begin{minipage}[b]{0.485\textwidth}
    \includegraphics[width=\textwidth,trim={0.6cm 1.5cm 0.8cm 1.5cm},
    clip]{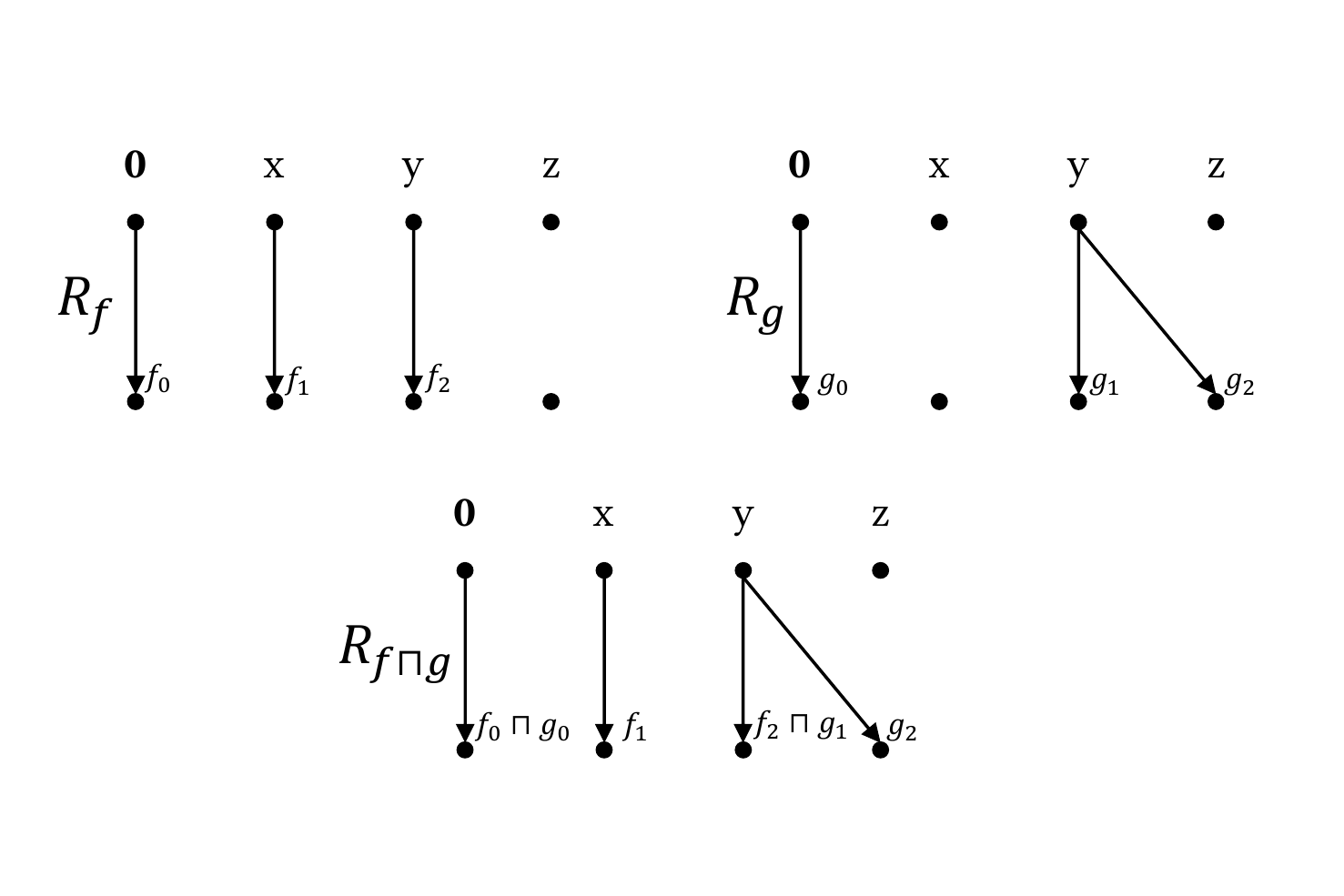}
  \end{minipage}
  ~ %
  \begin{minipage}[b]{0.485\textwidth}
    \includegraphics[width=\textwidth,trim={0.6cm 1.5cm 0cm 1.5cm},
    clip]{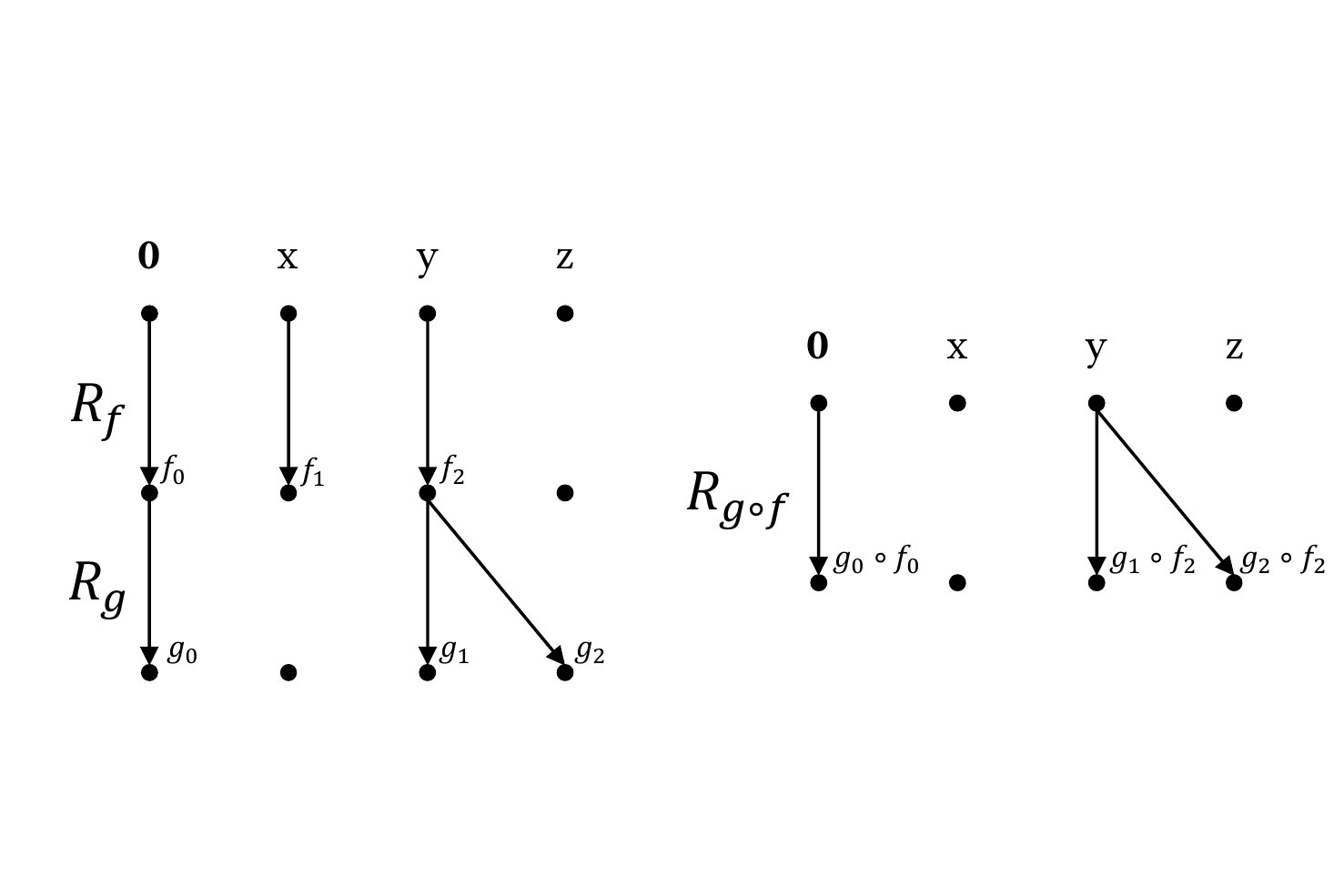}
  \end{minipage}
  \caption{Representation relations are closed under meet (left) and composition
    (right).
    In IFDS, the edges of the bipartite graphs are unlabeled; in IDE, they are
  labeled with micro-functions.}\label{fig:RepRelMeetCompose}
    \vspace{-10pt}
\end{figure}

\medskip

\noindent 
As discussed informally in \cref{sec:motivation},
we can encode event handling in the supergraph by modeling an event
loop that nondeterministically calls all event handlers. Such an encoding is
sound but imprecise, because it ignores the order in which event handlers are
called and admits infeasible paths that include handling of events before the
handler has been registered or the event has been emitted.

\medskip\noindent\textbf{IDE background.}
The IDE framework~\cite{Sagiv96} generalizes IFDS to interprocedural
distributive environment problems, in which dataflow facts are
  \emph{environments}, \ie, maps in $D \to L$ from a finite set~$D$ to a
  finite-height lattice~$L$, and dataflow functions are \emph{environment
  transformers} in~\mbox{$(D \to L) \to (D \to L)$} that \emph{distribute} over the
  meet operator of the map lattice $D \to L$. In other words, environments are
  values from the map lattice $D \to L$, which is lifted from the lattice $L$:
  the top element is $\TopEnv = \lambda d . \top$ where $\top$ is the top
  element of $L$, and for two environments $\env_1, \env_2$ in $D \to L$,
  \mbox{$\env_1 \sqcap \env_2 = \lambda d . ( \env_1(d) \sqcap \env_2(d) )$}.

Formally, an IDE problem instance is defined as $P = \langle \Gstar, D, L, \MEnv
\rangle$, where:

\begin{enumerate}
  \item $\Gstar = \pair{\Nstar}{\Estar}$ is the supergraph of the input
    program;
  \item $D$ is a finite set of program symbols, \eg, variables;
  \item $L$ is a finite-height lattice with top element $\top$; and
  \item $\MEnv : \Estar \to ((D \to L) \to (D \to L))$ is a function that assigns
    environment transformers to supergraph edges.
\end{enumerate}

\noindent
IDE computes the meet-over-valid-paths, \mbox{$\SolveIDE : \Nstar
\to (D \to L)$}, of the environment transformers, similar to IFDS. At each node
$n$ in the supergraph, IFDS computes only the presence or absence of each
element $d$ of the dataflow domain; however, IDE computes for each $d$ an
element $l$ of the lattice $L$. Thus, IFDS is a special case of IDE in which $L$
is fixed to be the two-point lattice, with $\top$ indicating absence and $\bot$
indicating presence of $d$. Intuitively, one can think of the IDE algorithm as
computing facts in $D$ that hold along interprocedurally valid paths while
simultaneously propagating and computing values from $L$ along those paths.
Formally, the meet-over-valid-paths solution is defined as
\begin{equation*}\label{eqn:solve-ide}
  \SolveIDE(P) = \lambda n . \!\bigsqcap_{p \in\, \VP(n)} \!\!\MEnv(p)(\TopEnv)
\end{equation*}
where $\MEnv$ is extended so that $\MEnv([e_1 \ldots e_k]) =
\MEnv(e_k) \circ \cdots \circ \MEnv(e_1) \circ \textsf{id}$.

An IDE dataflow function in $(D \to L) \to (D \to L)$, \ie, a distributive
environment transformer, can be encoded as a \emph{pointwise representation},
using a bipartite graph with $2(D + 1)$ nodes. The nodes are the same as in an
IFDS representation relation, but each edge $d' \to d$ is labeled by
$f_{d',d}$, a function in $L \to L$ called a \emph{micro-function}. By
distributivity, such a set of micro-functions is sufficient to represent an
environment transformer $t$, since \mbox{$t(\env)(d) = f_{\Zero,d}(\top) \sqcap
  \bigsqcap_{d' \in D} f_{d',d}(\env(d'))$}.

Pointwise representations are also closed under meet and composition, as shown
in \cref{fig:RepRelMeetCompose}. The meet
of two representations $R_{f}$ and $R_{g}$ is the union of edges
of $R_{f}$ and $R_{g}$, where the micro-function for a shared edge in
$R_{f \sqcap g}$ is the meet of the two micro-functions of that edge in
$R_{f}$ and $R_{g}$. The composition of two representations is
computed by connecting the two graphs and composing micro-functions along paths
in the resulting graph. Therefore, an instantiation of the IDE framework
requires an efficient representation of micro-functions as well as an efficient
implementation of their composition, meet, and equality test.

The IDE algorithm represents a given problem instance as a \emph{labeled
exploded supergraph} $\Ghash_{P} = \pair{\Nhash}{\Ehash}$, with each edge
$\pair{m}{d_1} \to \pair{n}{d_2}$ labeled by a micro-function
$f : L \to L$. The labels are given by a function $\EdgeFn : \Ehash \to (L \to
L)$. To compute the meet-over-valid-paths solution over the labeled exploded
supergraph, the IDE algorithm requires two phases. The first phase is similar to
IFDS, iteratively composing bipartite graphs for control-flow paths of
increasing length; this determines which nodes $\pair{n}{d}$ are
reachable. The second phase applies the composed micro-functions to determine,
for each node $\pair{n}{d}$, the value $l \in L$ that $d$ is mapped to.

\medskip

\noindent In our approach, we take the IFDS exploded supergraph as input and
produce an IDE labeled exploded supergraph by assigning micro-functions to
exploded supergraph edges. For a program with a single event handler, we use the
lattice $L$ to keep track of the event handler registrations and event emissions
that have taken place on each control-flow path. To support multiple event
handlers, we use the map lattice $H \to L$, where $H$ is the set of event
handlers in the program and $L$ is the lattice for a single event handler. This
allows us to track the registration and event emission for each event handler in
the program.

\section{Technique}\label{sec:technique}

\noindent Our technique is a transformation $\Transform : \Ghash \to
\pair{\Ghash}{\EdgeFn}$ of an arbitrary instance of the IFDS
analysis framework into an instance of the IDE analysis framework. The IDE
solution encodes the same dataflow facts as the IFDS solution, except that it
excludes dataflow facts reachable only along infeasible paths.

The input to our technique, an instance of the IFDS framework, is expressed as
an exploded supergraph $\Ghash$, which encodes the ICFG
of the program under analysis, the dataflow analysis, and the transfer
functions for that analysis. The output of our technique, an instance of the IDE
framework, is a labeled exploded supergraph $\pair{\Ghash}{\EdgeFn}$ where
$\EdgeFn$ assigns micro-functions in $L \to L$ to each edge of the exploded
supergraph.

The key idea of our transformation is to augment the exploded supergraph with an
encoding of event handler operations. We do this by encoding event
handler operations as micro-functions on the edges of the exploded supergraph.
Our technique does not change the nodes or edges of the exploded
supergraph; it only assigns micro-functions to the edges of that graph.
Therefore, it does not change the ICFG, the base
dataflow analysis, or its transfer functions.

Intuitively, an IFDS analysis asks which elements $d \in D$ are present at node
$n$ of the supergraph, while an IDE analysis asks what lattice
value $l \in L$ is associated with element $d \in D$ at node $n$. In our
technique, the lattice $L$ encodes event handler state: if an element $d$ at
node $n$ maps to an infeasible event handler state, then we conclude that at
node $n$, $d$ should be excluded from the results.

By solving this IDE instance, we achieve the effect of eliminating dataflow
facts that are reachable only along infeasible paths. In the rest of this
section, we describe how we encode event handler operations as micro-functions,
and how we transform an IDE solution back to an IFDS solution. We also discuss
theoretical properties of our technique.

\subsection{Representing event handler state}\label{sec:technique-event-state}

\begin{figure}[!t]
  \centering
  \begin{minipage}[b]{0.65\textwidth}
    \centering
    \includegraphics[width=\textwidth,trim={2cm 0.3cm 1cm 1cm},
    clip]{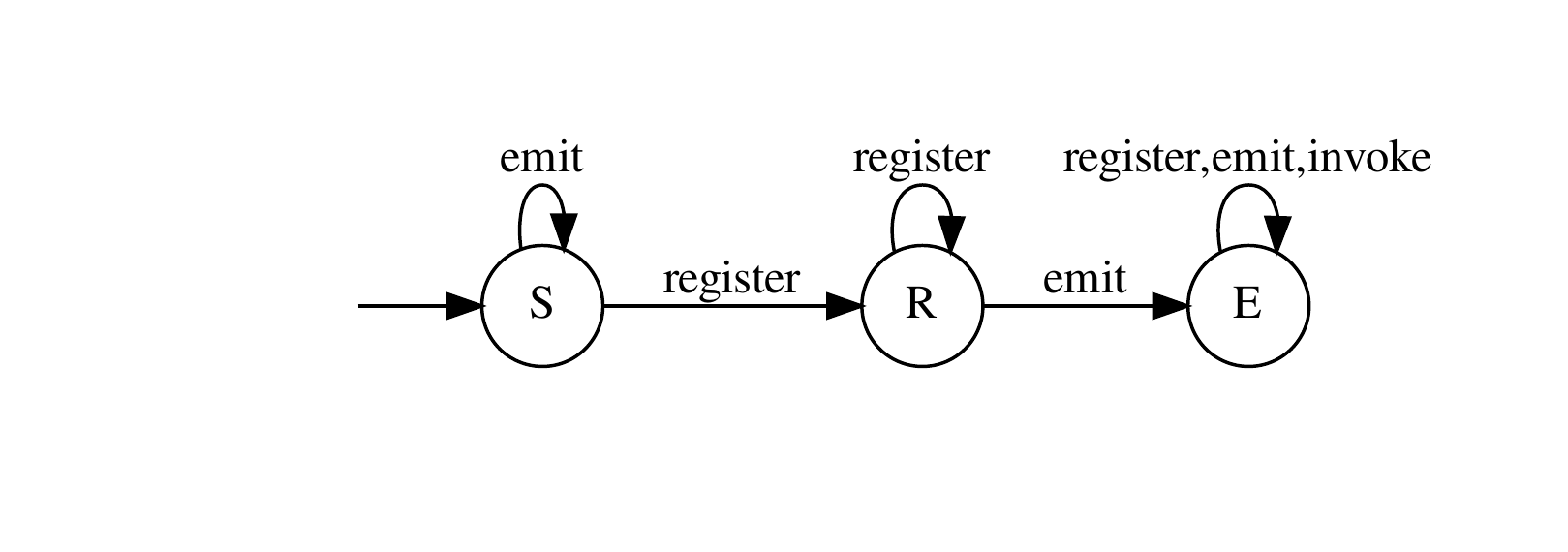}
  \end{minipage}
  ~ %
  \begin{minipage}[b]{0.3\textwidth}
    \centering
    \includegraphics[width=0.25\textwidth,trim={1.2cm 1.2cm 1.2cm 1.2cm},
    clip]{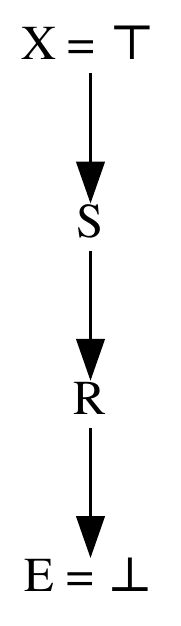}\\
  \end{minipage}
  \vspace{-5pt}
  \caption{Event handler state: concrete states and their transitions (left)
  and the lattice $L$ representing abstract states (right).}\label{fig:HandlerState}
    \vspace{-15pt}
\end{figure}

\noindent For simplicity of presentation, we restrict our attention in this
subsection to programs with a single event handler. We generalize to multiple
event handlers in the next subsection. We define three possible \emph{states}
for an event handler:

\begin{description}
  \item[S (start):] the event handler has not yet been registered.
  \item[R (registered):] the event handler has been registered for the event, but
    the event has not yet been emitted after registration. (Events emitted
    before registration are ignored.)
  \item[E (emitted):] the event handler has been registered and the event has
    been emitted after registration.
\end{description}

\noindent These states model the event handler during an actual program
execution. They are distinct from the event handler \emph{operations} (event
handler registration, event emission, and event handler invocation) we discussed 
in \cref{sec:motivation}, which cause transitions between the states. For
example, an event handler is initially in the \emph{start}~(S) state. When the
handler is registered, then its state becomes the \emph{registered}~(R) state.
When an event associated with that handler is emitted, the state becomes the
\emph{emitted}~(E) state. Only in this state can the handler be invoked from the
event loop; the handler can never be invoked from any other state. These
transitions are summarized in \cref{fig:HandlerState}.

To model this state machine in a static analysis, we need a fourth state,
\emph{infeasible}~(X). Invoking the event handler from the \emph{start}~(S)
state (before handler registration) or \emph{registered}~(R) state
(before event emission) can never happen at run time, but
such an ordering may arise during the analysis, so we must identify it as an
infeasible path. We use the IDE algorithm to keep track of event handler state
and rule out data flow along infeasible paths.

Specifically, we define $L$ to be the chain lattice over the set $\{X, S, R,
E\}$ with the ordering $X \sqsupseteq S \sqsupseteq R \sqsupseteq E$, as
depicted in \cref{fig:HandlerState}. The
lattice elements $S$, $R$, and $E$ indicate the corresponding states of the
event handler, and the top element $X$ indicates the infeasible state, \ie, the
dataflow fact has traversed a control-flow path that was infeasible.

Recall that the IDE algorithm maps a dataflow fact $d$ to the top element of $L$
to indicate that the fact does \emph{not} hold at the given program point. The
ordering between the four elements is designed to model the behavior at
control-flow merge points: when two control-flow paths merge, the associated
event handler state after the merge is the lesser of the two states before the
merge. For example, if one control-flow path has passed through an
\emph{infeasible sequence of operations}~(X) and the second control-flow path
has passed through a feasible sequence of operations that results in the event
handler being \emph{registered but not emitted}~(R), then after the control-flow
merge, the event handler is in state $X \sqcap R = R$; it may have been
\emph{registered but not emitted}~(R).

At the main entry point of the program, the event handler is defined to be in
the \emph{start}~(S) state for each fact $d$ that holds at the entry
point.\footnote{Normally, the IDE algorithm initializes every fact $d$ to the
  top element, \ie, $X$. In this case, we could label edges leaving the entry point
with micro-functions that update every fact $d$ to S.
However, for convenience, we simply initialize every fact $d$ to
S.} As dataflow facts are propagated during the analysis,
we track event handler state with IDE micro-functions,
encoding the state machine transitions along each edge of the
exploded supergraph. The default micro-function along most edges
is the identity, indicating that the event handler state
does not change. The other micro-functions are defined in \cref{tbl:microfun}
and correspond to the operations discussed in \cref{sec:motivation}.

\begin{wraptable}{r}{0.3\textwidth}
  \vspace{-30pt}
  \caption{Micro-function definitions.}\label{tbl:microfun}
  \vspace{10pt}
  \begin{tabular}{c| c c c}
      & $\register$ & $\emit$ & $\invoke$\\
    \hline
    X & X & X & X\\
    S & R & S & X\\
    R & R & E & X\\
    E & E & E & E
  \end{tabular}
  \vspace{-20pt}
\end{wraptable}

\medskip\noindent\textbf{Event handler registration.}
The first micro-function labeled edge is a control-flow edge that represents an
event handler registration operation. For example, the 
control-flow edge from \code{door.on('open', ...)} to the library
in \cref{fig:DoorsExample} causes the event handler to transition from the start
state to the registered state.
If the event handler is in any other state, then the
registration is ignored. We define the micro-function for this edge in
\cref{tbl:microfun}, first column.

\medskip\noindent\textbf{Event emission.}
The second micro-function labeled edge is a control-flow edge that represents
event emission. An example is the edge from
\code{door.emit('open')} to the library: the handler associated with the
\code{open} event transitions from the registered state to the emitted state.
In all other cases, the event emission is ignored. We define the micro-function
for this edge in \cref{tbl:microfun}, second column.

\medskip\noindent\textbf{Event handler invocation.}
The third micro-function labeled edge is a control-flow edge that represents
event handler invocation. Examples of these edges are from the library to
the start nodes of both event handlers. If the handler is \emph{not} in the
emitted state, then it transitions to the infeasible state because the
handler is being invoked before it has been registered or its event has been
emitted. We define the micro-function for this edge in
\cref{tbl:microfun}, third column.

\medskip\noindent\textbf{Discussion.}
The transformation~$\Transform : \Ghash \to \pair{\Ghash}{\EdgeFn}$
converts an instance of the IFDS framework to an instance of the IDE framework.
It does not change the structure of the exploded supergraph, $\Ghash$, but it
provides $\EdgeFn : \Ehash \to (L \to L)$, an assignment of exploded supergraph
edges to micro-functions. For programs with a single event handler,
$\EdgeFn$ is defined as follows:
\begin{equation*}
  \EdgeFn(e) =
  \begin{cases}
    \register & \quad \text{if edge } e \text{ registers the handler} \\
    \emit     & \quad \text{if edge } e \text{ emits an event for the
    handler} \\
    \invoke   & \quad \text{if edge } e \text{ invokes the handler from
    the event loop} \\
    \identity & \quad \text{otherwise}.
  \end{cases}
\end{equation*}

Returning to our example in \cref{fig:DoorsExample}, consider the execution
path that is actually taken at run time: the door opening event handler is
registered by \code{door.on('open', ...)} on line~\ref{line:RegisterHandleOpen},
the door opening event is emitted by \code{door.emit('open')} on
line~\ref{line:EmitOpen}, and the door opening event handler is invoked by the
edge from the library to \texttt{\small start}\textsubscript{\code{open}}.

For this control-flow path, the analysis computes the composition of the
micro-functions, namely $\invoke \circ \emit \circ \register$. Applying this
composed function to the initial state, we have $\invoke(\emit(\register(S))) =
E$, so any data flow associated with this path is considered feasible.

On the other hand, consider a control-flow path in which the event handler is
registered and invoked, but the event is never emitted. The composed
micro-function for such a path is $\invoke \circ \register$, so we have
$\invoke(\register(S)) = X$. Thus, any data flow computed along that path is
considered infeasible.

Recall that an instantiation of the IDE framework requires an efficient
representation of micro-functions and an efficient implementation of their
composition, meet, and equality test. A micro-function $f: L\to L$ can be
efficiently represented as a table of the four values
\mbox{\microfun{f(X)}{f(S)}{f(R)}{f(E)}}. Since there are only $4^4 = 256$ possible
such functions, compositions and meets of micro-functions can be precomputed,
and only 8 bits are required to represent a micro-function.

\subsection{Multiple event handlers}\label{sec:technique-multiple-handlers}

For programs with multiple events and multiple event handlers, it is necessary
for the analysis to distinguish them. In \cref{fig:DoorsExample}, the
control-flow path that registers the \code{hdlOpen} event handler, emits the
\code{open} event, and invokes the opening event handler is feasible. However,
the path that instead invokes the \code{hdlClose} handler should be
infeasible, because the door closing event handler has never been registered and
the \code{close} event has never been emitted. Our solution is to maintain a
separate state for each event handler.

Thus, we define the IDE lattice $L'$ to be the map lattice $H \to L$, where $H$
is the set of event handlers in the program and $L$ is the lattice for a single
event handler that we discussed in the previous subsection. For each node
$\pair{n}{d}$ in the exploded supergraph, the IDE algorithm using lattice $L'$
computes a map $m : H \to L$ that assigns a separate state for each
event handler in the program.

Recall that the IDE framework requires an efficient representation of
micro-functions in $L' \to L'$, which in this case is $(H \to L) \to (H \to L)$.
Efficiently representing such functions is non-trivial. There are
$(4^{|H|})^{4^{|H|}}$ possible functions of this type, so any representation
that could encode all of them would require $\Omega(|H|\cdot4^{|H|})$ bits to
encode each one. The key to an efficient encoding is the observation that all of
the micro-functions that actually occur during an analysis, including their
compositions and meets, are separable, in that the effect of an operation on the
state of one event handler is independent of the states of other event handlers
before the operation. In other words, the state that an event handler
transitions to depends only on that handler's previous state, and not the state
of any other event handler.

Each separable micro-function can thus be represented by a function in \mbox{$H
\to (L \to L)$} that models the effect $L \to L$ of an operation on each event
handler in $H$ separately. We discussed in the previous subsection how to
efficiently represent a function in $L \to L$. Now, to represent a
micro-function in $H \to (L \to L)$, we need only to tabulate $|H|$ functions of
type $L \to L$, one for each event handler in $H$. The operations by the IDE
framework, composition, meet, and equality comparison, are computed pointwise,
separately for each event handler. Effectively, a micro-function in $L' \to L'$
is represented by a map of event handlers to micro-functions in $L \to L$. Note
that this representation of micro-functions and the required operations adds a
factor of $O(|H|)$ to the asymptotic complexity of the IDE algorithm.

The version of $\EdgeFn : \Ehash \to (L' \to L')$ that supports multiple event handlers is therefore
defined as:
\begin{equation*}
  \EdgeFn(e) =
  \begin{cases}
    \register_{h} & \quad \text{if edge } e \text{ registers handler } h \\
    \emit_{h}     & \quad \text{if edge } e \text{ emits an event for handler }
      h \\
    \invoke_{h}   & \quad \text{if edge } e \text{ invokes handler } h
      \text{ from the event loop} \\
    \identity     & \quad \text{otherwise}.
  \end{cases}
\end{equation*}
We use the subscript $h$ to indicate that a micro-function updates only the
state assigned to $h$, and not the state of any other event handler.
(Note that the default micro-function, $\identity$, does not update \emph{any} state.)
In an implementation,
$\EdgeFn$ must also be able to determine which handler $h$ is affected by each
edge in the exploded supergraph.

\subsection{Transforming IDE results to IFDS results}\label{sec:untransform}

When IDE finishes analyzing a program, its output is, for each
program point, a map from elements of $D$ to elements of $L'$. To convert this
output to a result for the original IFDS problem, we must identify, at each
program point, the subset of elements of $D$ that are reachable along feasible
paths. In our context, a path is feasible if, for every event handler, the
operations affecting that event handler along the path are in a feasible
sequence (\eg, the handler is not invoked before it is registered or its event emitted). In other
words, a path is feasible if the element of $L'$ computed by the IDE analysis
maps every handler to a state other than X. Formally, we define an
``untransform'' function $\Untransform : (\Nstar \to (D \to L')) \to
(\Nstar \to D)$ that converts an IDE result $\mathcal{R}$ to an IFDS result:
\begin{equation*}
  \Untransform(\mathcal{R}) = \lambda n . \{ d \;|\; \forall h \in H \,.\,
    \mathcal{R}(n)(d)(h) \neq X \}.
\end{equation*}

In \cref{fig:DoorsExample}, on the control-flow path that first passes
through \code{door.on('open', ...)} and then through \code{door.emit('open')},
the micro-function for that path is $\{ h_\textsf{open} \mapsto \emit
\circ \register, \allowbreak h_\textsf{close} \mapsto \identity \}$, which
computes the event handler state mapping $\{ h_\textsf{open} \mapsto E,
\allowbreak h_\textsf{close} \mapsto S \}$. If that path then continues into
\code{hdlOpen}, the event state will remain at $\{ h_\textsf{open} \mapsto E,
\allowbreak h_\textsf{close} \mapsto S \}$, and thus the analysis will conclude
that the path is feasible.

However, if the path continues into \code{hdlClose} instead, the composed
micro-function becomes $\{ h_\textsf{open} \mapsto \emit \circ \register,
  \allowbreak h_\textsf{close} \mapsto \invoke \}$, which computes the event
  handler state mapping $\{ h_\textsf{open} \mapsto E, \allowbreak
  h_\textsf{close} \mapsto X\}$. Since at least one handler is in state X, the
  analysis will conclude that this path is infeasible and discard all dataflow
  facts computed along this path.

\subsection{Theoretical results}\label{sec:theoretical-results}

\smallskip\noindent\textbf{Soundness and precision.}
Our transformation is sound: the IDE analysis considers all feasible dataflow
paths, \ie, the ones that occur during a program execution. Any dataflow fact
that IFDS computes along a concrete path will be returned by our technique.

\begin{theorem}[Soundness]
  Let $P$ be an IFDS problem, $p = [\startmain, \ldots, n]$ be a concrete
  execution path, and $d \in D$ be a dataflow fact. Then:
  \begin{equation*}
    d \in \MF(p)(\varnothing)
    \implies
    d \in \Untransform\big(\SolveIDE(\Transform(P))\big)(n).
  \end{equation*}
\end{theorem}

\noindent Our transformation is precise: the IDE analysis returns a subset of
the dataflow facts that would be computed by IFDS. Dataflow facts computed along
infeasible paths are not included in the result of our transformation.

\begin{theorem}[Precision]
  Let $P$ be an IFDS problem and $n \in \Nstar$ be any node in the supergraph.
  Then:
  \begin{equation*}
    \Untransform\big(\SolveIDE(\Transform(P))\big)(n) \subseteq \SolveIFDS(P)(n).
  \end{equation*}
\end{theorem}

\smallskip\noindent\textbf{Efficiency.}
As discussed by Reps et al.~\cite[sec.~5]{Reps95},
the asymptotic complexity of solving an IFDS problem instance is \mbox{$O(|E|
\!\cdot\! |D|^3)$}. An equivalent IDE problem instance
also requires \mbox{$O(|E| \!\cdot\! |D|^3)$} time to
solve, provided that the micro-functions have an \emph{efficient
representation}~\cite[def.~5.2]{Sagiv96}. Our representation of micro-functions
adds a time and space overhead of $O(|H|)$. Therefore, the asymptotic complexity
of the event-driven IDE analysis is \mbox{$O(|E| \!\cdot\! |D|^3 \!\cdot\! |H|)$}.
\section{Implementation}\label{sec:implementation}

\noindent
To demonstrate the effectiveness of our technique on small-scale event-driven
programs, we implemented a proof-of-concept called \Borges,
which analyzes a subset of \js.

\subsection{Uninitialized variables analysis as an IFDS problem}\label{sec:implementation-ifds}

As input, \Borges takes a list of \js files to be analyzed (including a model of
any library functions used) and an \emph{event model specification} describing
which function calls represent event handler registrations, event emissions, and
event handler invocations. \Borges transforms the
IFDS problem into an IDE problem, solves the IDE problem, and filters out results
that were computed by traversing infeasible paths.

\Borges is implemented as a Scala application and builds on two program
analysis infrastructures: TAJS~\cite{Jensen09} and Flix~\cite{Madsen16}. We use
TAJS to construct control flow graphs and call graphs for \js programs. \Borges
uses the control flow graph as the basis for constructing the supergraph that is
used by IFDS and IDE, and the call graph to determine which
functions are invoked from each call site. We use Flix to solve the IFDS and IDE
problems; in particular, we implement the analyses in the Flix
language and instantiate the uninitialized variables analysis by implementing
the dataflow functions in Scala. In principle, however, \Borges is applicable to
any programming language and dataflow problem that can be expressed in the IFDS
framework.

One challenge that we encountered involves the handling of arrays and objects.
In \js, arrays are list-like objects that may be non-contiguous, and object
properties are accessed via string values that may be computed at run time,
posing significant challenges to static
analysis~\cite{DBLP:conf/ecoop/SridharanDCST12}.  Since the challenge of
precisely modeling objects and arrays is largely orthogonal to the issue of
avoiding infeasible paths in the presence of event-handling constructs, we chose
to adopt a simplistic approach where the abstract locations used to represent
objects and arrays are unified with those representing their elements. In other
words, if an object (array) is initialized, then so are all its properties
(elements).

\clearpage
\subsection{Transforming to an IDE problem}\label{sec:implementation-ide}

In order to produce more precise results, \Borges transforms IFDS problems
into IDE problems that track the operations associated with each event
handler, as well as each handler's state. Information about which function calls
correspond to which event handler operations must be provided to \Borges as an
\emph{event model specification}, which also indicates the
argument that represents the event name and the argument that represents
the event handler. Using this information, \Borges can identify which call sites
involve event handler operations.

For example, the program in \cref{fig:DoorsExample} uses the Node.js
\code{events} library. Applying static analysis to complex libraries poses challenges
that are beyond the scope of this paper, and our approach to handle library-based applications
is to provide a stub that models the library's essential functions and control flow.  
In the stub for the \code{events} library, we provide the functions \code{on},
\code{emit}, and \code{_eventDispatcher}. The event model specifies that a call
to \code{on} (\eg, \code{on('open', hdlOpen)}) \emph{registers} the second
argument (\code{hdlOpen}) as an event handler on the event given as the first
argument (\code{open}), a call to \code{emit} (\eg, \code{emit('open')})
\emph{emits} the event given as its argument (\code{open}), and a call from
inside the library (specifically, from \code{_eventDispatcher}) \emph{invokes}
an event handler.

Using this information, along with the output from TAJS, \Borges constructs a
mapping of event handler registrations that happen in a program. For each edge in the
control flow graph, \Borges can identity whether it affects event handler
state (\ie, through a registration, event emission, or invocation), and if so, which event
name and event handler is involved. Furthermore, \Borges also computes a mapping
from event names to event handlers, to easily identify which handler responds to
a given event emission.

The transformation from an IFDS problem to an IDE problem is straightforward.
Recall that the IFDS algorithm uses an exploded supergraph to represent dataflow
functions, while in the IDE algorithm, $\EdgeFn$ assigns a micro-function to
exploded supergraph edges. \Borges provides such an implementation of
$\EdgeFn$ to determine the micro-function for a given edge and
event handler. For instance, the edge representing a call to
\code{register('open', hdlOpen)} is labeled with the $\register$
micro-function for the \code{hdlOpen} handler.

With all the exploded supergraph edges labeled, solving the IDE problem
computes the composition of all the micro-functions along a control-flow
path, taking the meet whenever multiple paths merge. In other words, when
computing dataflow facts for the possibly uninitialized variables analysis,
\Borges also maintains the event handler states. Thus, before reporting a final
result for each program point, \Borges can examine the states of each event
handler and filter out any result with an event handler in the infeasible state.
\clearpage
\section{Case Studies}\label{sec:casestudy}

In this section, we discuss three examples to demonstrate our approach. We
return to the file system example in \cref{sec:motivation} and briefly discuss
two other programs. We run \Borges on three small, event-driven Node.js
applications, and apply our transformation to a possibly uninitialized variables
analysis.

\medskip\noindent\textbf{File system module, revisited.}
Recall \cref{fig:dirstat}, where
\code{sum} is read without being initialized, but only along an infeasible
path.
\Borges can improve precision by considering the order in which
callbacks are executed. Specifically, the calls to \code{readdir}
(line~\ref{line:readdir}) and \code{stat} (line~\ref{line:stat}) are
\emph{registration} operations for the \code{f} and \code{h} callbacks,
respectively. However, the \emph{emission}
operation is implicit and happens from within the event loop. Since event emission
happens after event handler registration but before event handler
invocation, we model it as occurring immediately after registration. In other
words, the micro-function labeling the calls to \code{readdir} and \code{stat}
is $\emit \circ \register$. Finally, invocations of
\code{f} and \code{h} are \emph{invocation} operations, which correspond to the
micro-function $\invoke$.

When \Borges analyzes the application, it identifies two paths with
respect to the callbacks. In one path, \code{readdir} is called, \code{f} is
invoked, \code{stat} is called, and \code{h} is invoked.
The composition of micro-functions along this path is
$\{ h_\textsf{f} \mapsto \invoke \circ \emit \circ \register, h_\textsf{h}
\mapsto \invoke \circ \emit \circ \register \}$, which computes the event
handler state mapping $\{ h_\textsf{f} \mapsto E, h_\textsf{h} \mapsto E \}$,
meaning the path is feasible.

However, in the infeasible path where \code{readdir} is called and then \code{h}
is invoked, the composed micro-function is $\{ h_\textsf{f} \mapsto \emit \circ
  \register, h_\textsf{h} \mapsto \invoke \}$, which computes the event handler
  state $\{ h_\textsf{f} \mapsto E, h_\textsf{h} \mapsto X \}$, meaning the
  path is infeasible. Therefore, any results computed along this path are
  filtered out.

\begin{wrapfigure}{r}{0.5\textwidth}
  \vspace{-30pt}
  \centering
  {\scriptsize 
\begin{lstlisting}
/*@
% Keep this in sync with timer2.js, the copy is a hack to avoid duplicate labels
@*/console.log('Enter a number to start the timer.'); /*@\label{line:askNumber}@*/
var stdin = process.openStdin(); /*@\label{line:openStdin}@*/
var rem; /*@\label{line:declareRemaining}@*/
stdin.on('data', function start(sec) { /*@\label{line:stdinCallback}@*/
    rem = sec; /*@\label{line:initializeRemaining}@*/
    setTimeout(tick, 1000); /*@\label{line:callSetTimeout1}@*/
});
function tick() { /*@\label{line:setTimeoutCallback}@*/
    rem = rem - 1; /*@\label{line:readRemaining}@*/
    console.log(rem); /*@\label{line:printRemaining}@*/
    if (rem > 0) { /*@\label{line:checkRemaining}@*/
        setTimeout(tick, 1000); /*@\label{line:callSetTimeout2}@*/
    } else {
        process.exit(0); /*@\label{line:processExit}@*/
    }
}
\end{lstlisting}}
  \vspace{-10pt}
  \caption{Example application using timers.}\label{fig:timer}
  \vspace{-20pt}
\end{wrapfigure}
\medskip\noindent\textbf{Timers module.}
\Cref{fig:timer} implements a simple
timer. It is similar to the file system example, as it has two callbacks
that can be executed only in a certain order. The application
prompts the user for a number and then counts down from that
number in one-second intervals. It uses the \code{timers}
module,
 whose
functions are defined in the global scope.

Because the callbacks \code{start} (line \ref{line:stdinCallback}) and
\code{tick} (line~\ref{line:setTimeoutCallback}) are invoked asynchronously, a
traditional static analysis might consider an execution path where
\code{tick} is executed before
\mbox{\code{start},} and conclude that \code{rem} is possibly uninitialized when it is
read on line~\ref{line:readRemaining}. However, this is an infeasible
path: \code{tick} is only registered as a callback by \code{start} and itself,
so it can be invoked only after \code{start} has finished executing.
As a result, \Borges labels the execution path with
the micro-function $\{
  h_\textsf{start} \mapsto \emit \circ \register, h_\textsf{tick} \mapsto
  \invoke \}$ and computes the event handler state mapping as $\{
h_\textsf{start} \mapsto E, h_\textsf{tick} \mapsto X \}$.

\begin{wrapfigure}{r}{0.5\textwidth}
  \vspace{-10pt}
  \centering
  {\scriptsize 
\begin{lstlisting}
/*@
% Keep this in sync with server.js, the copy is a hack to avoid duplicate labels
@*/var net = require('net'); /*@\label{line:requireNet}@*/
var nConn; /*@\label{line:declareNrConnects}@*/
var svr = net.createServer(); /*@\label{line:createServer}@*/
svr.listen(8080, function lstn() {/*@\label{line:serverListen}@*/
    svr.on('connection',
      function conn(cxn) { /*@\label{line:serverConnected}@*/
        console.log('client connected');/*@\label{line:printClientConnect}@*/
        nConn++; /*@\label{line:incrementNrConnects}@*/
        console.log('connects: '+nConn);/*@\label{line:printNrConnects}@*/
        cxn.pipe(cxn); /*@\label{line:pipeConnection}@*/
      });
    console.log('server is lstn'); /*@\label{line:printServerListening}@*/
    nConn = 0; /*@\label{line:initializeNrConnects}@*/
});
\end{lstlisting}}
  \vspace{-10pt}
  \caption{Example application using \code{net}.}\label{fig:server}
  \vspace{-20pt}
\end{wrapfigure}
\medskip\noindent\textbf{Net module.}
The program in \cref{fig:server} implements a small TCP server
using the Node.js \code{net}
module.
It creates a server that
listens for client connections and mirrors input back to the client. A
corresponding client application could be implemented in \js using the
\code{net} module, or in any other language of choice.

Without an ordering constraint between the \code{lstn}
(line~\ref{line:serverListen}) and \code{conn}
(line~\ref{line:serverConnected})
callbacks, a traditional analysis might consider infeasible paths, \eg,
where \code{conn} is invoked before \code{lstn}.
Along this path, the analysis concludes
that \code{nConn} on line~\ref{line:incrementNrConnects} is
possibly uninitialized. However, \code{conn} can be
executed only after \code{lstn} finishes, which guarantees that
\code{nConn} is initialized. In \Borges, such a path would be labeled by the
micro-function $\{ h_\textsf{lstn} \mapsto \emit \circ \register,
h_\textsf{conn} \mapsto \invoke \}$, which computes the event handler state
$\{ h_\textsf{lstn} \mapsto E, h_\textsf{conn} \mapsto X \}$.

\section{Related work}\label{sec:RelatedWork}

Bodden et al.~\cite{DBLP:conf/pldi/BoddenTRBBM13} use the IDE algorithm to
enhance the precision of an IFDS analysis when analyzing software product lines.
 They modify any
IFDS analysis into an IDE analysis that runs on the original program and tracks
the product line variants in which each dataflow fact holds. 

Rapoport et al.~\cite{Rapoport15} observe that context-sensitive
analysis can be made more precise by correlating the
dynamic dispatch behavior of different call sites on the same receiver object.
They also transform an arbitrary IFDS analysis into an IDE analysis that keeps
track of which methods have been dynamically dispatched on each receiver.

Jhala and Majumdar~\cite{DBLP:conf/popl/JhalaM07} adapt IFDS for
\emph{asynchronous programs}. In these programs, asynchronous calls are similar
to event registrations in that the procedure will be invoked at a later time;
however, there are no event emissions, so the time of invocation is
unpredictable. In their approach, instead of encoding additional state as an IDE
problem, they transform the analysis into a larger IFDS analysis that tracks, at
each asynchronous call site, the number of pending asynchronous calls made for
which the procedure has not yet been invoked.

Madsen et al.~\cite{Madsen15} introduce the \emph{event-based
call graph}, an extension of the call graph that models happens-before
constraints between event handler registrations and event emissions. However,
their approach does not scale well because the number of contexts is exponential
in the size of the program. 

Sotiropoulos et al.~\cite{Sotiropoulos19} introduce $\lambda_q$, a model
of asynchrony in \js, as well as the \emph{callback graph}, which describes the
possible orderings of callback execution. They design a
\emph{callback-sensitive} analysis for \js that uses the callback graph to
respect the execution order of callbacks. Their technique is specific to \js,
while our approach is language agnostic.

\section{Conclusion}\label{sec:Conclusion}

Traditional static analyses produce imprecise results when applied to
event-driven programs because they assume that event handler callbacks can
execute in any order.  We have presented an approach for precise dataflow
analysis that is applicable to any dataflow problem that can be expressed as
an instance of the IFDS framework, and is expressed as a transformation from
that presentation to an IDE problem, where the dataflow functions associated
with edges in the graph filter out infeasible paths that arise due to
impossible sequences of event handler invocations. We prove the correctness
of our transformation and report on a proof-of-concept tool.

\bibliographystyle{splncs04}
\bibliography{bib/refs}

\begin{thebibliography}{10}
\providecommand{\url}[1]{\texttt{#1}}
\providecommand{\urlprefix}{URL }
\providecommand{\doi}[1]{https://doi.org/#1}

\bibitem{DBLP:conf/pldi/ArztRFBBKTOM14}
Arzt, S., Rasthofer, S., Fritz, C., Bodden, E., Bartel, A., Klein, J., {Le
  Traon}, Y., Octeau, D., McDaniel, P.: {FlowDroid: Precise Context, Flow,
  Field, Object-sensitive and Lifecycle-aware Taint Analysis for Android Apps}.
  In: Proc. {ACM} {SIGPLAN} Conference on Programming Language Design and
  Implementation, PLDI (2014). \doi{10.1145/2594291.2594299}

\bibitem{Heros}
Bodden, E.: {Heros IFDS/IDE Solver}. \url{https://github.com/Sable/heros},
  accessed: 2018-10-05

\bibitem{DBLP:conf/pldi/BoddenTRBBM13}
Bodden, E., Tol{\^{e}}do, T., Ribeiro, M., Brabrand, C., Borba, P., Mezini, M.:
  {SPL\({}^{\mbox{LIFT}}\): Statically Analyzing Software Product Lines in
  Minutes Instead of Years}. In: Proc. {ACM} {SIGPLAN} Conference on
  Programming Language Design and Implementation, PLDI (2013).
  \doi{10.1145/2491956.2491976}

\bibitem{Fink:2008:ETV:1348250.1348255}
Fink, S.J., Yahav, E., Dor, N., Ramalingam, G., Geay, E.: {Effective Typestate
  Verification in the Presence of Aliasing}. ACM Transactions on Software
  Engineering and Methodology, TOSEM  \textbf{17}(2),  9:1--9:34 (2008).
  \doi{10.1145/1348250.1348255}

\bibitem{WALA}
{IBM Research}: {Watson Libraries for Analysis (WALA)}.
  \url{https://github.com/wala/WALA}, accessed: 2018-10-05

\bibitem{Jensen09}
Jensen, S.H., M{\o}ller, A., Thiemann, P.: {Type Analysis for JavaScript}. In:
  Proc. Static Analysis Symposium, SAS (2009).
  \doi{10.1007/978-3-642-03237-0\_17}

\bibitem{DBLP:conf/popl/JhalaM07}
Jhala, R., Majumdar, R.: {Interprocedural Analysis of Asynchronous Programs}.
  In: Proc. {ACM} {SIGPLAN-SIGACT} Symposium on Principles of Programming
  Languages, POPL (2007). \doi{10.1145/1190216.1190266}

\bibitem{Madsen15}
Madsen, M., Tip, F., Lhot{\'{a}}k, O.: {Static Analysis of Event-Driven
  {Node.js} {JavaScript} Applications}. In: Proc. {ACM} {SIGPLAN} Conference on
  Object-Oriented Programming, Systems, Languages, and Applications, OOPSLA
  (2015). \doi{10.1145/2814270.2814272}

\bibitem{Madsen16}
Madsen, M., Yee, M.H., Lhot\'{a}k, O.: {From Datalog to \textsc{Flix}: A
  Declarative Language for Fixed Points on Lattices}. In: Proc. {ACM} {SIGPLAN}
  Conference on Programming Language Design and Implementation, PLDI (2016).
  \doi{10.1145/2908080.2908096}

\bibitem{DBLP:conf/oopsla/NaeemL08}
Naeem, N.A., Lhot{\'{a}}k, O.: {Typestate-like Analysis of Multiple Interacting
  Objects}. In: Proc. {ACM} {SIGPLAN} Conference on Object-Oriented
  Programming, Systems, Languages, and Applications, OOPSLA (2008).
  \doi{10.1145/1449764.1449792}

\bibitem{Rapoport15}
Rapoport, M., Lhot{\'{a}}k, O., Tip, F.: {Precise Data Flow Analysis in the
  Presence of Correlated Method Calls}. In: Proc. Symposium on Static Analysis,
  SAS (2015). \doi{10.1007/978-3-662-48288-9\_4}

\bibitem{Reps95}
Reps, T., Horwitz, S., Sagiv, S.: {Precise Interprocedural Dataflow Analysis
  via Graph Reachability}. In: Proc. {ACM} {SIGPLAN-SIGACT} Symposium on
  Principles of Programming Languages, POPL (1995). \doi{10.1145/199448.199462}

\bibitem{Sagiv96}
Sagiv, S., Reps, T., Horwitz, S.: {Precise Interprocedural Dataflow Analysis
  with Applications to Constant Propagation}. In: Proc. Conference on Theory
  and Practice of Software Development, CAAP/FASE (1995).
  \doi{10.1007/3-540-59293-8\_226}

\bibitem{Sotiropoulos19}
Sotiropoulos, T., Livshits, B.: {Static Analysis for Asynchronous JavaScript
  Programs}. In: Proc. European Conference on Object-Oriented Programming,
  ECOOP (2019). \doi{10.4230/LIPIcs.ECOOP.2019.8}

\bibitem{DBLP:conf/ecoop/SridharanDCST12}
Sridharan, M., Dolby, J., Chandra, S., Sch{\"{a}}fer, M., Tip, F.: Correlation
  tracking for points-to analysis of {JavaScript}. In: Proc. European
  Conference on Object-Oriented Programming, ECOOP (2012).
  \doi{10.1007/978-3-642-31057-7\_20}

\end{thebibliography}

\clearpage\appendix
\section{Proofs}\label{sec:appendix-proofs}

Our work is based on the work by Rapoport et al.~\cite{Rapoport15}, which also transforms a given IFDS
problem instance to an IDE analysis that eliminates dataflow facts computed
along infeasible paths.

In this section, we assume that $\langle \Gstar, D, F, \MF, \sqcap \rangle$ and
its exploded supergraph representation, $\Ghash = \langle \Nhash, \Ehash
\rangle$, is the base IFDS problem instance given to our transformation, and that
$\langle \Gstar, D, L', \MEnv \rangle$ and its labeled exploded supergraph
representation, $\pair{\Ghash}{\EdgeFn}$, is the IDE problem instance
defined in \cref{sec:technique}; in particular, lattice $L'$ is the map
lattice $H \to L$ where $L$ is the event handler state lattice.
Finally, to simplify some notation, we write the edge $e_i = n_{i-1} \to
n_i$ for each $i$. Note that $e_1 = \startmain \to n_1$.

\subsection{Soundness and Precision}

Recall that in the IDE definition, we used $\TopEnv$ to denote the top element
of the environment lattice, \ie, the environment $\lambda d . \top$ that maps
every element to $\top$. We also defined the meet-over-valid-paths solution for
an IDE problem as $\SolveIDE(P) = \lambda n . \bigsqcap_{p \in\; \VP(n)}
\MEnv(p)(\TopEnv)$. However, for the event-driven analysis, the initial state is
$\SEnv = \lambda d . S$ rather than $\TopEnv$. Thus, the meet-over-valid-paths
solution for the event-driven analysis is:
\begin{equation*}
  \SolveIDE(P) = \lambda n . \!\bigsqcap_{p \in\; \VP(n)} \!\!\MEnv(p)(\SEnv).
\end{equation*}

To prove the soundness and precision theorems, we require two lemmas.

\begin{lemma}\label{lemma:a}
  Let $p = [ \startmain, \ldots, n ]$ be a concrete execution trace of some
  program, and let $h \in H$ be an event handler in the program. If at node $n$
  of the trace $p$, handler $h$ is in state $q$, and $d \in D$ is a dataflow
  fact such that $d \in \MF(p)(\varnothing)$, then $q \sqsupseteq
  \MEnv(p)(\SEnv)(d)(h)$.
\end{lemma}

Intuitively, the lemma states that the event-driven analysis over-approximates
event handler state in a program execution. Note that $q$ is a concrete state,
so it cannot be X.

\begin{proof}
  By induction on the length of the program trace.\\

  \noindent
  \emph{Base case:} $p = [ \startmain ]$. There is no instruction (edge)
  in the trace, so there is no dataflow fact $d$. Therefore, the lemma trivially
  holds.\\

  \noindent
  \emph{Induction hypothesis:} Let $p = [ \startmain, \ldots, n_k ]$ and let
  $\hat{q} = \MEnv(p)(\SEnv)(d_k)(h)$, \ie, $\hat{q}$ is the abstract state
  computed by the event-driven analysis for the execution trace $p$, $d_k$ is
  some dataflow fact in $\MF(p)(\varnothing)$, and $h$ is some event handler.
  Suppose the lemma holds for trace $p$, \ie, $q \sqsupseteq \hat{q}$ where $q$
  is the concrete state for handler $h$ at node $n_k$ after the trace $p$.\\

  \noindent
  \emph{Induction step:} Now consider $p' = [ \startmain, \ldots, n_k, n_{k+1}
  ]$. Let $q'$ be the concrete state for handler $h$ at node $n_{k+1}$ after the
  trace $p'$. We must now show $q' \sqsupseteq \MEnv(p')(\SEnv)(d)(h)$.\\

  \noindent
  Because $\MEnv$ is extended from edges to paths by composition, we can
  rewrite:
  \begin{align*}
    \MEnv(p')(\SEnv)(d) &= \big( \MEnv(e_{k+1}) \circ \MEnv(e_k) \circ \cdots \circ
                      \MEnv(e_1) \big)(\SEnv)(d) \\
                      &= \big( \MEnv(e_{k+1})\big(\MEnv(e_k)( \cdots
                    (\MEnv(e_1)(\SEnv)) \cdots )\big)\big)(d) \\
                    &= \MEnv(e_{k+1})\big( \MEnv(p)(\SEnv) \big)(d).
  \end{align*}
  Note that $\MEnv(p)(\SEnv)$ computes the environment at node $n_k$ after the
  trace $p$, which is then transformed by $\MEnv(e_{k+1})$ to get the
  environment at node $n_{k+1}$, a single node after the trace $p$, which is a
  map from $D \to (H \to L)$. Thus, $\MEnv(e_{k+1})\big(\MEnv(p)(\SEnv)\big)(d)$
  returns a map from handlers to event handler states.\\

  \noindent
  Now, recall that for a given environment $\env : D \to L'$, the IDE framework
  represents an environment transformer $t : (D \to L') \to (D \to L')$ as a set
  of micro-functions in $L' \to L'$:
  \begin{equation*}
    t(\env)(d) = f_{\Zero,d}(\top) \sqcap \bigsqcap_{d' \in D}\!
      f_{d',d}\big(\env(d')\big).
  \end{equation*}
  For an edge $n_1 \to n_2 \in \Estar$, $\MEnv(n_1 \to n_2)$ gives the
  environment transformer for that edge, and for $d_1, d_2 \in D \cup \{ \Zero
  \}$, $\EdgeFn\big(\pair{n_1}{d_1} \to \pair{n_2}{d_2}\big)$ gives the
  corresponding micro-functions:
  \begin{align*}
    & \MEnv(n_1 \to n_2)(\textsf{env})(d) \\
      &\qquad = \EdgeFn\big(\langle n_1, \Zero \rangle \to
      \langle n_2, d \rangle\big)(\top) \, \sqcap \, \bigsqcap_{d' \in D}\!
      \EdgeFn\big( \langle n_1, d' \rangle \to \langle n_2, d
      \rangle\big)\big(\textsf{env}(d')\big).
  \end{align*}

  By substitution, we can rewrite:
  \begin{align*}
    & \MEnv(e_{k+1})\big( \MEnv(p)(\SEnv) \big)(d)(h) \\
    &\qquad = \bigg( \EdgeFn\big(\langle n_k, \Zero \rangle \to \langle n_{k+1},
        d \rangle\big)(\top) \, \sqcap \,\\
        &\qquad\qquad \bigsqcap_{d' \in D}\! \EdgeFn\big( \langle n_k, d' \rangle
        \to \langle n_{k+1}, d \rangle\big)\big( \MEnv(p)(\SEnv)(d')
        \big) \bigg)\!(h) \\
    &\qquad \sqsubseteq \left( \, \bigsqcap_{d' \in D}\! \EdgeFn\big( \langle
        n_k, d' \rangle \to \langle n_{k+1}, d \rangle\big)\big(
        \MEnv(p)(\SEnv)(d') \big) \right) \!\!(h) \\
    &\qquad \sqsubseteq \EdgeFn\big( \langle n_k, d_k \rangle \to \langle
        n_{k+1}, d \rangle\big)\big( \MEnv(p)(\SEnv)(d_k) \big)(h)
  \end{align*}
  This gives us the inequality:
  \begin{equation*}
    \EdgeFn\big( \langle n_k, d_k \rangle \to \langle n_{k+1}, d \rangle\big)
    \big( \MEnv(p)(\SEnv)(d_k) \big)(h) \sqsupseteq \MEnv(p')(\SEnv)(d)(h).
  \end{equation*}
  The inequality compares two different ways of computing the state of handler
  $h$ for dataflow fact $d$ at node $n_{k+1}$ (after the trace $p'$). On the
  right-hand side, the entire environment at node $n_{k}$ (after the trace $p$)
  is transformed by $\MEnv(e_{k+1})$, and then the state of handler $h$ is
  obtained from the new environment. On the left-hand side, at node $n_k$ (after
  the trace $p$), a map of event handlers to states (\ie, an element of the
  lattice $L' = H \to L$), is obtained for some dataflow fact $d_k$ and then
  updated by the micro-function $\EdgeFn\big(\pair{n_k}{d_k} \to
  \pair{n_{k+1}}{d} \big)$, before getting the state mapped to handler $h$. The
  inequality states that the left-hand side is more precise than the right-hand
  side; intuitively, this is because the left-hand side takes the effect of a
  single micro-function, while the right-hand side takes the effect of merging all
  the micro-functions.\\

  \noindent
  It remains to show $q' \sqsupseteq \EdgeFn\big(\pair{n_k}{d_k} \to
  \pair{n_{k+1}}{d}\big)\big( \MEnv(p)(\SEnv)(d_k) \big)(h)$ to complete the
  proof. To simplify notation, let $m = \MEnv(p)(\SEnv)(d_k)$ be the map of
  event handlers to states, as computed by the IDE algorithm along path $p$ for
  dataflow fact $d_k$. Note that $\hat{q} = m(h)$. We proceed by considering the
  four cases of $\EdgeFn$ and how the micro-functions update the map $m$.

    \begin{case}
      $e_{k+1}$ is an edge that registers handler $h$, so the
      micro-function is $\register_h$.\\

      \noindent
      The micro-function for this edge updates the state for handler $h$: if $h$
      is in state S, then $h$ will be in state R. Otherwise, the state is
      unchanged.
      The concrete state of handler $h$ at node $n_k$ is state $q$, which cannot
      be X, so there are three possibilities:
      \begin{itemize}
        \item If $q = S$, then edge $e_{k+1}$ registers handler $h$, so we get
          the new concrete state $q' = R$. By the induction hypothesis, $q
          \sqsupseteq m(h)$, so at node $n_k$, $h$ is mapped to S, R, or E. In
          each of those cases, $R \sqsupseteq \register_h(m)(h)$, so the lemma
          holds.

        \item If $q = R$, then the event handler has already been registered, so
          the state is unchanged and $q' = R$. By the induction hypothesis, $q
          \sqsupseteq m(h)$, so at node $n_k$, $h$ is mapped to R or E. In both
          of those cases, $R \sqsupseteq \register_h(m)(h)$, so the lemma holds.

        \item If $q = E$, then the event handler has already been registered
          (and its event has been emitted), so the state is unchanged and $q' =
          E$. By the induction hypothesis, $q \sqsupseteq m(h)$, so at node
          $n_k$, $h$ is mapped to E. In this case, $\register_h(m) = m$, so $E
          \sqsupseteq \register_h(m)(h)$, and the lemma holds.
      \end{itemize}
    \end{case}

    \begin{case}
      $e_{k+1}$ is an edge that emits an event for handler $h$, so
      the micro-function is $\emit_h$.\\

      \noindent
      The micro-function for this edge updates the state for handler $h$: if $h$
      is in state R, then $h$ will be in state E. Otherwise, the state is
      unchanged.
      The concrete state of handler $h$ at node $n_k$ is state $q$, which cannot
      be X, so there are three possibilities:
      \begin{itemize}
        \item If $q = S$, then the event emission is ignored, so $q' = S$. By
          the induction hypothesis, \mbox{$q \sqsupseteq m(h)$}, so at node
          $n_k$, $h$ is mapped to S, R, or E. In each of those cases, \mbox{$S
          \sqsupseteq \emit_h(m)(h)$}, so the lemma holds.

        \item If $q = R$, then the handler can respond to the event, so we get
          the new concrete state $q' = E$. By the induction hypothesis, $q
          \sqsupseteq m(h)$, so at node $n_k$, $h$ is mapped to R or E. In both
          of those cases, $E \sqsupseteq \emit_h(m)(h)$, so the lemma holds.

        \item If $q = E$, then the state is unchanged, so $q' = E$. By the
          induction hypothesis, $q \sqsupseteq m(h)$, so at node $n_k$, $h$
          is mapped to E. In this case, $\emit_h(m) = m$, so $E \sqsupseteq
          \emit_h(m)(h)$, and the lemma holds.
      \end{itemize}
    \end{case}

    \begin{case}
      $e_{k+1}$ is an edge \emph{from} the event loop to handler
      $h$, so the micro-function is $\invoke_h$.\\

      \noindent
      The micro-function for this edge updates the state for handler $h$: if $h$
      is in state E, then the state is unchanged. Otherwise, the state will be
      X.
      The concrete state of handler $h$ at node $n_k$ is state $q$, which cannot
      be X, S, or R. X never occurs during a concrete execution. S is not
      possible because it means the event handler has not been registered, so
      invocation cannot occur. R is not possible because it means the event has
      not been emitted, so invocation cannot occur. Therefore, $q = q' = E$.
      By the induction hypothesis, $q \sqsupseteq m(h)$, so at node $n_k$, $h$
      is mapped to E. In this case, $\invoke_h(m) = m$, so $E \sqsupseteq
      \emit_h(m)(h)$, and the lemma holds.
    \end{case}

    \begin{case}
      $e_{k+1}$ is any other edge, so the micro-function is
      $\identity$.\\

      \noindent
      The micro-function does not update the state of handler $h$. Similarly, in
      the concrete execution, there is no event handler operation on this edge,
      so $q' = q$. By the induction hypothesis, $q \sqsupseteq m(h)$, and
      $\identity(m) = m$, so $q' \sqsupseteq \identity(m)(h)$
      and the lemma holds. \qed
  \end{case}
\end{proof}

\begin{lemma}\label{lemma:b}
  Let $p = [ \startmain, \ldots, n ]$ be a concrete execution trace of
  some program, $h \in H$ be an event handler, and $d \in D$ be a dataflow fact.
  Then:
  \begin{equation*}
    d \in \MF(p)(\varnothing) \iff \MEnv(p)(\SEnv)(d)(h) \neq X.
  \end{equation*}
\end{lemma}

Intuitively, the lemma states that for a concrete execution path, the
event-driven analysis never computes an infeasible event handler state.

\begin{proof}

  \hspace*{\fill}\\

  \noindent \textit{$\implies$ direction}. By induction on the length of the
  program trace.\\

  \noindent
  \emph{Base case:} $p = [ \startmain ]$. There is no instruction (edge) in the
  trace, so there is no dataflow fact $d$. Therefore, the lemma trivially
  holds.\\

  \noindent
  \emph{Induction hypothesis:} Let $p = [ \startmain, \ldots, n_k ]$ and let
  $\hat{q} = \MEnv(p)(\SEnv)(d_k)(h)$, \ie, $\hat{q}$ is the abstract state
  computed by the event-driven analysis for the execution trace $p$, $d_k$ is
  some dataflow fact in $\MF(p)(\varnothing)$, and $h$ is some event handler.
  Suppose the lemma holds for trace $p$, \ie, $d \in \MF(p)(\varnothing)
  \implies \hat{q} \neq X$.\\

  \noindent
  \emph{Induction step:} Now consider $p' = [ \startmain, \ldots, n_k, n_{k+1}
  ]$. Let $q'$ be the concrete state for handler $h$ at node $n_{k+1}$ after the
  trace $p'$. We must now show $d \in \MF(p')(\varnothing) \implies
  \MEnv(p')(\SEnv)(d_k)(h) \neq X$.\\

  \noindent
  From the previous proof, we know:
  \begin{align*}
    \MEnv(p')(\SEnv)(d)(h) \sqsubseteq
    \EdgeFn\big( \langle n_k, d_k \rangle \to \langle n_{k+1}, d
        \rangle\big)\big( \MEnv(p)(\SEnv)(d_k) \big)(h).
  \end{align*}
  By the induction hypothesis, $\MEnv(p)(\SEnv)(d_k)(h) \neq X$ for all $h$, so
  we know that
  $\MEnv(p)(\SEnv)(d_k)$ is a map $m$ where each handler is mapped to S, R, or
  E. So we need to examine $m'$, the map $m$ after being updated by the
  micro-function on edge $\pair{n_k}{d_k} \to \pair{n_{k+1}}{d_{k+1}}$.\\

  \noindent
  Of the four cases, three of them ($\register_\textsf{h}$, $\emit_\textsf{h}$,
  and $\identity$) are straightforward. None of these micro-functions map any
  handler to X. So, for all $h \in H$, we have:
  \begin{align*}
    \MEnv(p')(\SEnv)(d)(h) \sqsubseteq m'(h).
  \end{align*}
  Therefore, $\MEnv(p')(\SEnv)(d)(h) \neq X$.\\

  \noindent
  The fourth case is when $\EdgeFn$ returns $\invoke_\textsf{h}$, which will map
  $h$ to X, \emph{unless} handler $h$ is currently mapped to E. However, along
  the concrete execution trace $p'$, the last edge $n_k \to n_{k+1}$ corresponds
  to an invocation of event handler $h$. This can only happen if $h$ has already
  been registered and its event emitted. In other words, the concrete state of
  $h$ must be E. By \cref{lemma:a}, $E \sqsupseteq m(h)$ so $m(h) = E$ and
  $\invoke_h(m)(h) = E$. Therefore, $\MEnv(p')(\SEnv)(d)(h) \neq X$.\\

  \medskip

  \noindent \textit{$\impliedby$ direction}.\\

  \noindent
  The premise states that after a concrete execution trace $p$, at node $n$ and
  dataflow fact $d$, handler $h$ is in a state other than X. In other words,
  there exists a path $p$ in the exploded supergraph to node $n$ where $d$
  holds, so by definition, $d \in \MF(p)(\varnothing)$. \qed
\end{proof}

We can now prove the soundness and precision theorems.

\setcounter{theorem}{0}

\begin{theorem}[Soundness]
  Let $P$ be an IFDS problem, $p = [\startmain, \ldots, n]$ be a concrete
  execution path, and $d \in D$ be a dataflow fact. Then:
  \begin{equation*}
    d \in \MF(p)(\varnothing)
    \implies
    d \in \Untransform\big(\SolveIDE(\Transform(P))\big)(n).
  \end{equation*}
\end{theorem}

\begin{proof}
  Recall the definitions of ``untransform'' $\Untransform(\mathcal{R})$ and
  meet-over-valid-paths for the event-driven analysis $\SolveIDE$:
  \begin{align*}
    \Untransform(\mathcal{R}) &= \lambda n . \{ d' \;|\; \forall h \in H \,.\,
      \mathcal{R}(n)(d')(h) \neq X \} \\
    \SolveIDE(P) &= \lambda n . \!\bigsqcap_{p' \in\; \VP(n)} \!\!\MEnv(p')(\SEnv).
  \end{align*}
  By substitution, we get:
  \begin{align*}
    \Untransform\big(\SolveIDE(\Transform(P))\big)(n) &= \left\{ d' \;|\;
      \forall h \in H \,.\, \SolveIDE(\Transform(P))(n)(d')(h) \neq X
      \right\} \\
    &= \left\{ d' \;|\; \forall h \in H \,.\, \!\!\left( \bigsqcap_{p'
       \in\, \VP(n)}\!\! \MEnv(p')(\SEnv) \!\right)\!\!(d')(h) \neq X \right\}\!.
  \end{align*}
  We have $d \in \MF(p)(\varnothing)$ so by \cref{lemma:b},
  $\MEnv(p)(\SEnv)(d)(h) \neq X$ for any $h \in H$. In other words, for the
  concrete path $p$, the event-driven analysis computes an environment $\env$
  where $\env(d)(h) \neq X$ for all handlers $h$. Such an environment is
  included in the meet-over-valid-paths $\bigsqcap_{p' \in\, \VP(n)}
  \MEnv(p')(\SEnv)$, whose result is a new environment $\env'$ where
  $\env'(d)(h) \neq X$ for all handlers $h$.\footnote{Recall our definition of
  environments in $D \to L'$, where $\env_1 \sqcap \env_2 = \lambda d . \lambda
  h . (\env_1(d)(h) \sqcap \env_2(d)(h))$} Therefore, $d \in
  \Untransform(\SolveIDE(\Transform(P)))(n)$. \qed
\end{proof}

\begin{theorem}[Precision]
  Let $P$ be an IFDS problem and $n \in \Nstar$ be any node in the supergraph.
  Then:
  \begin{equation*}
    \Untransform\big(\SolveIDE(\Transform(P))\big)(n) \subseteq \SolveIFDS(P)(n).
  \end{equation*}
\end{theorem}

\begin{proof}
  Let $P$ be an instance of the IFDS framework. On the right-hand side, we have:
  \begin{equation*}
    \SolveIFDS(P)(n) = \!\bigsqcap_{p \in\, \VP(n)} \!\!\MF(p)(\varnothing).
  \end{equation*}
  On the left-hand side, we have:
  \begin{align*}
    \Untransform\big(\SolveIDE(\Transform(P))\big)(n) &= \left\{ d \;|\;
      \forall h \in H \,.\, \SolveIDE(\Transform(P))(n)(d)(h) \neq X
      \right\} \\
    &= \left\{ d \;|\; \forall h \in H \,.\, \!\!\left( \bigsqcap_{p
       \in\, \VP(n)}\!\! \MEnv(p)(\SEnv) \!\right)\!\!(d)(h) \neq X \right\}\!.
  \end{align*}
  Consider a dataflow fact $d \in \Untransform\big( \SolveIDE(\Transform(P))
  \big)(n)$. This implies that for all event handlers $h \in H$, there exists at
  least one valid path $p' \in\, \VP(n)$ where $\MEnv(p')(\SEnv)(d)(h) \neq X$.
  (Otherwise, there exists a handler $h$ such that for all paths $p \in\,
  \VP(n)$, $\MEnv(p)(\SEnv)$ computes an environment $\env$ where
  $\env(d)(h) = X$, and taking the meet over all those environments implies $d$
  is not in our result.) By \cref{lemma:b}, this implies $d \in
  \MF(p')(\varnothing)$. Therefore, $d \in \bigsqcap_{p \in\, \VP(n)}
  \MF(p)(\varnothing)$, and so $d$ is a fact computed by IFDS. \qed
\end{proof}

\subsection{Efficiency}

According to Sagiv et al.~\cite[def.~5.2]{Sagiv96}, an IDE problem instance is
efficiently representable if its class of micro-functions $F \subseteq L \to L$
satisfies the following properties:

\begin{itemize}
  \item There is a representation for the identity and top functions.
  \item The representation is closed under meet and composition.
  \item $F$ is a finite-height lattice.
  \item Application, composition, meet, and equality can be computed in constant
    time.
  \item The representation of any function $f \in F$ requires constant space.
\end{itemize}

Recall that for the event-driven IDE analysis, our micro-functions are in $(H
\to L) \to (H \to L)$, where $H$ is the set of event handlers and $L$ is the
event state lattice. A function $f : L \to L$ is represented as
$\microfun{f(X)}{f(S)}{f(R)}{f(E)}$. To represent micro-functions in $(H \to L)
\to (H \to L)$, we actually represent them as functions in $H \to (L \to L)$,
because the state of one event handler does not affect the state of another
event handler. Intuitively, this representation is a map from event handlers in
$H$ to functions in $L \to L$.

Let us now consider whether our representation of functions in $H \to (L \to L)$
is efficient, according to the above properties:

\begin{description}
  \item[Representation for identity and top.]
    The identity function maps every event handler to the identity function in
    $L \to L$, \ie $\big\{ h \mapsto \microfun{X}{S}{R}{E} \big\}$ for all $h
    \in H$. The top function maps every event handler to the top function in $L
    \to L$, \ie $\big\{ h \mapsto \microfun{X}{X}{X}{X} \big\}$ for all $h \in
    H$.

  \item[Representation closed under meet and composition.]
    Consider two functions $f, g : L \to L$. Their representations are
    $\microfun{f(X)}{f(S)}{f(R)}{f(E)}$ and $\microfun{g(X)}{g(S)}{g(R)}{g(E)}$.

    $f \sqcap g$ is represented as $\microfun{f(X) \sqcap g(X)}{f(S) \sqcap
    g(S)}{f(R) \sqcap g(R)}{f(E) \sqcap g(E)}$.

    $g \circ f$ is represented as
    $\microfun{g(f(X))}{g(f(S))}{g(f(R))}{g(f(E))}$.

    Now consider two micro-functions $f', g' : H \to (L \to L)$. Their
    representations are $f' = \{ h_1 \mapsto f_1, \ldots, h_n \mapsto f_n \}$
    and $g' = \{ h_1 \mapsto g_1, \ldots, h_n \mapsto g_n \}$, where each $h_i
    \in H$ is a handler and each
    $f_i, g_i$ is a 4-tuple representation of a function in $L \to L$.

    $f' \sqcap g'$ is represented as $\big\{ h_1 \mapsto f_1 \sqcap g_1, \ldots,
    h_n \mapsto f_n \sqcap g_n \big\}$.

    $g' \circ f'$ is represented as $\big\{ h_1 \mapsto g_1 \circ f_1, \ldots,
    h_n \mapsto g_n \circ f_n \big\}$.

  \item[The micro-functions form a finite-height lattice.]
    First, let us consider functions in $L \to L$. $L$ is the event state
    lattice, which has only four elements. The representation of a function in
    $L \to L$ is effectively a 4-tuple $L \times L \times L \times L$, where the
    lattice ordering is defined pointwise, for each element of the tuple. That
    lattice has finite height; in fact, there are only $4^4 = 256$ elements in
    that lattice.

    Now, our micro-functions are represented as maps in $H \to (L \to L)$. These
    maps also form a lattice, where the ordering is pointwise for each $h \in
    H$. Since there are only finitely many event handlers in a program, there
    are only finitely many maps.

  \item[Time complexity of operations.]
    Application, composition, meet, and equality of functions in $L \to L$
    can be computed in constant time. However, our micro-functions are in $H
    \to (L \to L)$, which means operations need to be computed for each event
    handler in the program. Therefore, these operations require $O(|H|)$ time.

  \item[Space complexity of the representation.]
    The representation of a function in $L \to L$ requires constant
    space---in fact, the representation requires 8 bits. To represent
    functions in $H \to (L \to L)$, we require $|H|$ copies of each of the
    micro-function in $L \to L$. Therefore, the space requirement is
    $O(|H|)$.
\end{description}

The asymptotic complexity of the IDE algorithm is $O(|E| \!\cdot\! |D|^3)$,
assuming the micro-functions are efficiently representable. As discussed in the
IDE paper, the algorithm consists of a series of composition steps. There are at
most $O(|E| \!\cdot\! |D|^2)$ edges in the exploded supergraph, and each edge
can be used $O(|D|)$ times, so there are at most $O(|E| \!\cdot\! |D|^3)$
composition steps.

The IDE algorithm requires at most $O(|E| \!\cdot\! |D|^3)$ composition steps,
and since each step in our analysis costs $O(|H|)$, the overall time complexity
of the event-driven analysis is $O(|E| \!\cdot\! |D|^3 \!\cdot\! |H|)$.

\section{Case Studies in Detail}\label{sec:appendix-casestudy}

\begin{figure}[!ht]
  \vspace{-20pt}
  \centering
  \begin{minipage}[c]{0.5\textwidth}
    {\scriptsize 
\begin{lstlisting}
console.log('Enter a number to start the timer.'); /*@\label{line:askNumber2}@*/
var stdin = process.openStdin(); /*@\label{line:openStdin2}@*/
var rem; /*@\label{line:declareRemaining2}@*/
stdin.on('data', function start(sec) { /*@\label{line:stdinCallback2}@*/
    rem = sec; /*@\label{line:initializeRemaining2}@*/
    setTimeout(tick, 1000); /*@\label{line:callSetTimeout12}@*/
});
function tick() { /*@\label{line:setTimeoutCallback2}@*/
    rem = rem - 1; /*@\label{line:readRemaining2}@*/
    console.log(rem); /*@\label{line:printRemaining2}@*/
    if (rem > 0) { /*@\label{line:checkRemaining2}@*/
        setTimeout(tick, 1000); /*@\label{line:callSetTimeout22}@*/
    } else {
        process.exit(0); /*@\label{line:processExit2}@*/
    }
}
\end{lstlisting}}
  \end{minipage}
  \begin{minipage}[c]{0.49\textwidth}
    \includegraphics[width=\textwidth,trim={0.8cm 0.3cm 0.8cm 0.5cm},clip]{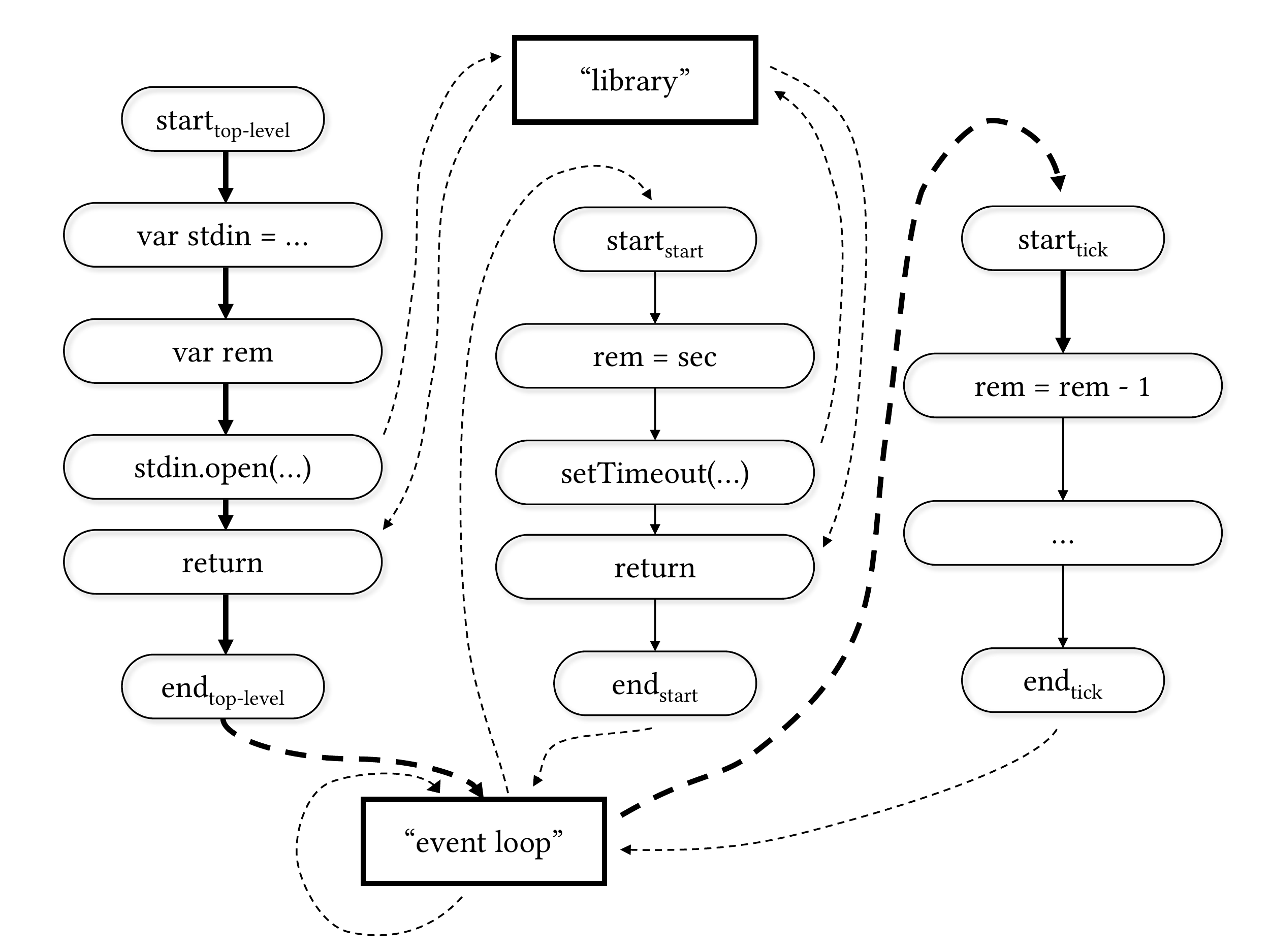}
  \end{minipage}
  \vspace{-5pt}
  \caption{\code{timer.js} from \cref{fig:timer} reproduced with its supergraph.
  Interprocedural edges are dashed; an infeasible path is shown in
  bold.}\label{fig:appendix-timer}
  \vspace{-15pt}
\end{figure}
\medskip\noindent\textbf{Timers module.}
When the application in \cref{fig:appendix-timer} executes, it first prints a message asking the user for
input (line~\ref{line:askNumber2}). Next, it connects a stream to standard input
(line~\ref{line:openStdin2}), making the stream available as an object assigned
to the variable \code{stdin}. The application also declares, but does not
initialize, the variable \code{rem} (line~\ref{line:declareRemaining2}). The call
to \code{stdin.on} (line~\ref{line:stdinCallback2}) takes two arguments,
registering the \code{start} event handler on the \code{data} event. The
\code{data} event is emitted when data is available on the standard input
stream, which causes the \code{start} handler to be invoked with the input data
as an argument. When this happens, the variable \code{rem} is initialized with
the user-provided input, \code{sec} (line~\ref{line:initializeRemaining2}).
Then, \code{setTimeout} is called (line~\ref{line:callSetTimeout12}), which
schedules a one-time execution of the \code{tick} callback after a one-second
delay. When \code{tick} is invoked, it decrements the \code{rem} variable
(line~\ref{line:readRemaining2}) and prints its value
(line~\ref{line:printRemaining2}). If \code{rem} is positive, then \code{tick}
calls \code{setTimeout} to schedule itself after another one-second delay
(line~\ref{line:callSetTimeout22}); otherwise, the application exits
(line~\ref{line:processExit2}).

In \Borges, the calls to \code{stdin.on} (line~\ref{line:stdinCallback2}) and
\code{setTimeout} (lines~\ref{line:callSetTimeout12}
and~\ref{line:callSetTimeout22}) are the \emph{registration} operations for the
\code{start} and \code{tick} callbacks, respectively. Similar to the \code{fs}
module, the \emph{emission} operations are implicit, so we model them as
occurring immediately after registration, and the \emph{invocation} operations
occur when \code{start} and \code{tick} are invoked from the event loop.
Therefore, when we run the analysis, the path that calls \code{stdin.on} and
invokes \code{tick} is labeled with the micro-function $\{ h_\textsf{start}
\mapsto \emit \circ \register, h_\textsf{tick} \mapsto \invoke \}$, and the
computed event handler state mapping is $\{ h_\textsf{start} \mapsto E,
h_\textsf{tick} \mapsto X \}$. In other words, this is an infeasible path, so
the possibly uninitialized variable on line~\ref{line:readRemaining2} should be
excluded.

\begin{figure}[!t]
  \centering
  \begin{minipage}[c]{0.5\textwidth}
    {\scriptsize 
\begin{lstlisting}
var net = require('net'); /*@\label{line:requireNet2}@*/
var nConn; /*@\label{line:declareNrConnects2}@*/
var svr = net.createServer(); /*@\label{line:createServer2}@*/
svr.listen(8080, function lstn() {/*@\label{line:serverListen2}@*/
    svr.on('connection',
      function conn(cxn) { /*@\label{line:serverConnected2}@*/
        console.log('client connected');/*@\label{line:printClientConnect2}@*/
        nConn++; /*@\label{line:incrementNrConnects2}@*/
        console.log('connects: '+nConn);/*@\label{line:printNrConnects2}@*/
        cxn.pipe(cxn); /*@\label{line:pipeConnection2}@*/
      });
    console.log('server is lstn'); /*@\label{line:printServerListening2}@*/
    nConn = 0; /*@\label{line:initializeNrConnects2}@*/
});
\end{lstlisting}}
  \end{minipage}
  \begin{minipage}[c]{0.49\textwidth}
    \includegraphics[width=\textwidth,trim={1.2cm 0.3cm 0.8cm 0.5cm},clip]{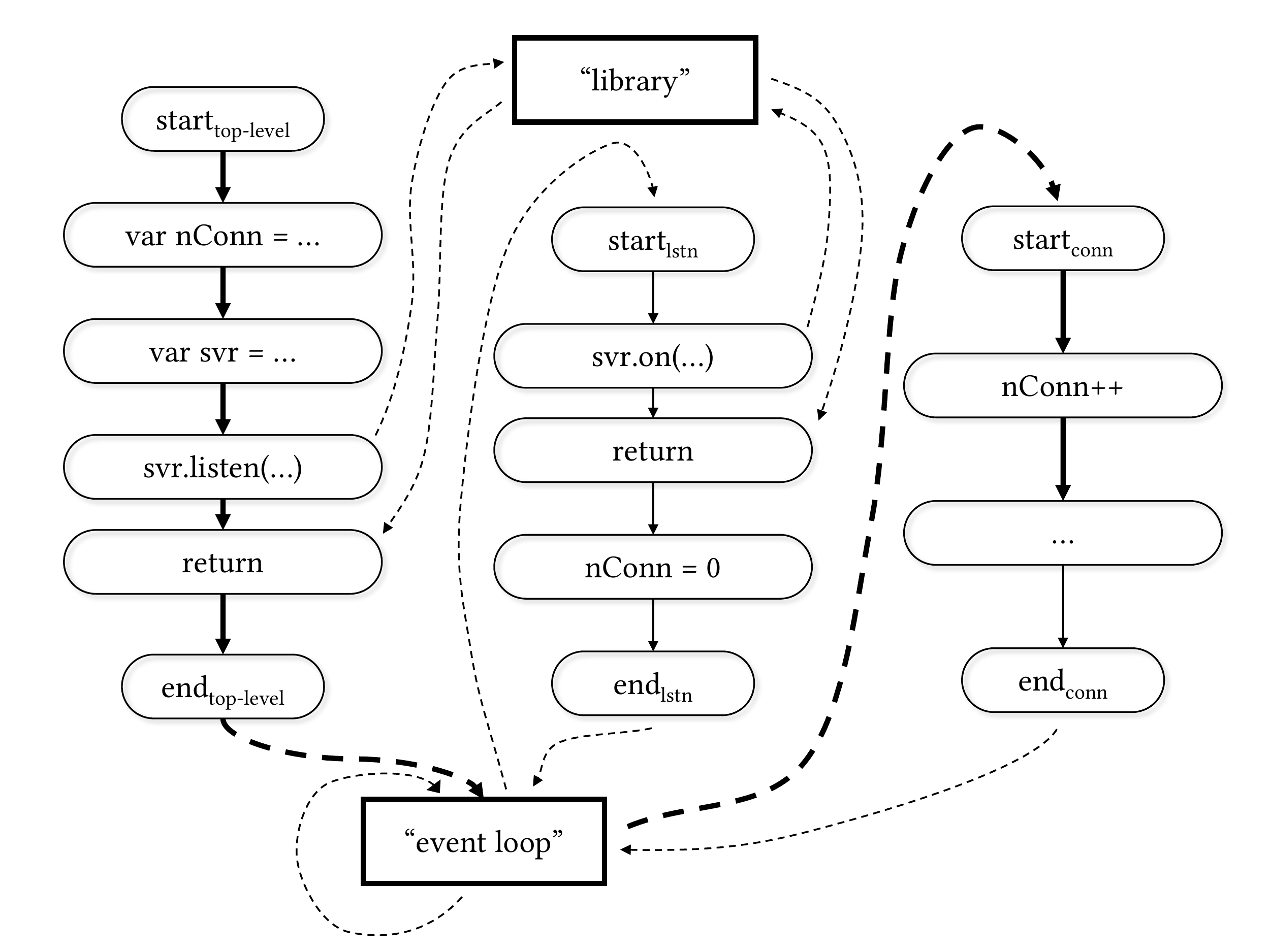}
  \end{minipage}
  \vspace{-5pt}
  \caption{\code{server.js} from \cref{fig:server} reproduced with its
  supergraph. Interprocedural edges are dashed; an infeasible path is shown in
  bold.}\label{fig:appendix-server}
  \vspace{-15pt}
\end{figure}
\medskip\noindent\textbf{Net module.}
The application in \cref{fig:appendix-server} starts by loading the \code{net} module
(line~\ref{line:requireNet2}), making its functions available as methods on
an object assigned to the variable \code{net}. Next, it declares, but does not
initialize, the variable \code{nConn} (line~\ref{line:declareNrConnects2}),
which counts the number of clients that have connected, and creates a server
(line~\ref{line:createServer2}), assigning the server object to the \code{svr}
variable. The call to \code{svr.listen} (line~\ref{line:serverListen2}) takes
two arguments: a port number for the server, and a callback \code{lstn}
that is invoked when the server is ready to accept client connections. When the
\code{lstn} callback is invoked, it calls \code{svr.on}
(line~\ref{line:serverListen2}) to register a second callback, \code{conn},
that is invoked on the \code{cxn} event, \ie, whenever a client connects.
Then, it prints a short message (line~\ref{line:printServerListening2})
indicating that the server is listening for connections, and initializes the
\code{nConn} counter to \code{0} (line~\ref{line:initializeNrConnects2}).
When a client connects, the \code{conn} callback is invoked with a
\code{cxn} object that contains information about the current connection.
The \code{conn} callback prints a message that a client has connected
(line~\ref{line:printClientConnect2}), increments the \code{nConn} counter
(line~\ref{line:incrementNrConnects2}), and prints the new value of
\code{nConn} (line~\ref{line:printNrConnects2}). Finally, the connection is
piped back to itself (line~\ref{line:pipeConnection2}), which has the effect of
mirroring input from the client back to the client.

To handle this situation in \Borges, we model the call to \code{svr.listen}
(line~\ref{line:serverListen2}) as a \emph{registration} of the \code{lstn}
event handler, with an implicit \emph{emission} occurring immediately
afterwards. Similarly, the call to \code{svr.on} is a registration of the
\code{conn} event handler on the \code{cxn} event, with an implicit
event emission from the library. When the callbacks are invoked from the event
loop, we model them as event handler \emph{invocations}. This gives us the same
micro-functions as our previous examples: $\emit \circ \register$ for the
\code{lstn} callback when \code{svr.listen} is called, $\emit \circ
\register$ for the \code{conn} callback when \code{svr.on} is called,
and $\invoke$ when \code{lstn} and \code{conn} are invoked. If we
consider the path where \code{conn} executes before \code{lstn} and
apply the composed micro-function to the start state, then we get $\{
h_\textsf{lstn} \mapsto E, h_\textsf{conn} \mapsto X \}$. Therefore,
this is an infeasible path that cannot occur at run time, so the uninitialized
variable on line~\ref{line:incrementNrConnects2} is a false positive.
 
\end{document}